\documentclass[11pt]{article}
 \expandafter\let\csname equation*\endcsname\relax
 \expandafter\let\csname endequation*\endcsname\relax
\usepackage{bm}
\usepackage{enumerate}
\usepackage[top=1in, bottom=1in, left=1in, right=1in]{geometry}
\usepackage{mathrsfs,color}
\usepackage{amsfonts}\usepackage{multirow}
\usepackage{amssymb}
\usepackage{dsfont}
\usepackage{caption,comment}
\usepackage{blkarray}
\usepackage{amsmath}
\usepackage{float}
\usepackage{amssymb,amsthm}
\usepackage[toc,page]{appendix}
\usepackage{graphicx,afterpage}
\usepackage{epstopdf}
\usepackage{cancel}
\usepackage{mathdots}
\usepackage[pdftex,colorlinks=true]{hyperref}
\usepackage[sort,compress]{cite}

\usepackage[textsize=tiny]{todonotes}

\usepackage{mhequ}
\hypersetup{CJKbookmarks,%
bookmarksnumbered,%
colorlinks,%
linkcolor=black,%
citecolor=black,%
plainpages=false,%
pdfstartview=FitH}
\numberwithin{equation}{section}
\numberwithin{figure}{section}
\def\theequation{\arabic{section}.\arabic{equation}}
\newcommand\tabcaption{\def\@captype{table}\caption}
\def\thefigure{\arabic{section}.\arabic{figure}}
\newtheorem{thm}{Theorem}[section]

\newtheorem{lem}[thm]{Lemma}
\newtheorem{prop}[thm]{Proposition}
\newtheorem{defn}[thm]{Definition}

\newtheorem{aspt}[thm]{Assumption}

\newcommand{\bfOne}{{\bf{1}}}

\newcommand{\bfx}{\mathbf{x}}
\newcommand{\bfK}{\mathbf{k}}
\newcommand{\bfI}{\mathbf{I}}
\newcommand{\bfz}{\mathbf{z}}
\newcommand{\bfA}{\mathbf{A}}
\newcommand{\bfJ}{\mathbf{J}}
\newcommand{\bfB}{\mathbf{B}}
\newcommand{\bfD}{\mathbf{D}}

\newcommand{\bfM}{\mathbf{M}}
\newcommand{\bfy}{\mathbf{y}}
\newcommand{\bfv}{\mathbf{v}}

\newcommand{\bfg}{\mathbf{g}}

\newcommand{\bfH}{\mathbf{H}}
\newcommand{\bfR}{\mathbf{R}}
\newcommand{\bfL}{\mathbf{L}}
\newcommand{\bfU}{\mathbf{U}}

\newcommand{\E}{\mathbb{E}}
\newcommand{\Prob}{\mathbb{P}}

\newcommand{\reals}{\mathbb{R}}

\newcommand{\unit}{\mathds{1}}

\newcommand{\calI}{\mathcal{I}}

\newcommand{\bfP}{\mathbf{P}}

\newcommand{\bfC}{\mathbf{C}}

\newcommand{\bfu}{\mathbf{u}}

\newcommand{\calD}{\mathcal{D}}
\newcommand{\calC}{\mathcal{C}}
\newcommand{\bfxtilde}{\tilde{\bfx}}
\newcommand{\bfztilde}{\tilde{\bfz}}
\newcommand{\bfvtilde}{\tilde{\bfv}}

\begin{document}
\title{
MALA-within-Gibbs samplers for
high-dimensional distributions with sparse conditional structure}
\author{X.T.~Tong, M.~Morzfeld and Y.M.~Marzouk}
\maketitle
\begin{abstract}
Markov chain Monte Carlo (MCMC) samplers are numerical methods for drawing samples from a given target probability distribution. We discuss one particular MCMC sampler, 
the MALA-within-Gibbs sampler, from the theoretical and practical perspectives.
We first show that the acceptance ratio and step size of this sampler are independent of 
the overall problem dimension
when (\emph{i}) the target distribution has sparse conditional structure,
and (\emph{ii}) this structure is reflected in the partial updating strategy of MALA-within-Gibbs.
If, in addition, the target density is block-wise log-concave, 
then the sampler's convergence rate is independent of dimension. 
From a practical perspective, we expect that MALA-within-Gibbs is useful for solving high-dimensional
Bayesian inference problems where the posterior exhibits sparse conditional structure at least approximately. 
In this context, a partitioning of the state that correctly reflects
the sparse conditional structure must be found, and we illustrate this
process in two numerical examples.
We also discuss trade-offs between the block size used for partial
updating and computational requirements that may increase with the
number of blocks. 
%
\end{abstract}

\begin{keywords}
Bayesian computation, high-dimensional distributions, Markov chain Monte Carlo
\end{keywords}

\section{Introduction}
Markov chain Monte Carlo (MCMC) samplers are numerical methods for drawing samples from an arbitrary target probability distribution whose density is known up to a normalizing constant. Generically, a Metropolis-Hastings MCMC sampler proposes a move by drawing from a proposal distribution and accepts or rejects the move with a probability that ensures that the stationary distribution of the Markov chain is the target distribution.

To design or chose a sampler for a given distribution, one typically considers the following three criteria. First, the type of proposal distribution is chosen based on how much information about the target distribution is available. For example, the Metropolis adjusted Langevin algorithm (MALA) requires derivatives of the target, while the random walk Metropolis (RWM) algorithm does not.  Second, the step size,
which controls how far the proposed sample strays from the current MCMC state,
needs to be tuned.
Put simply, too large a step size leads to poor mixing because the acceptance probability is too low; 
too small a step size leads to a large acceptance probability,
but the mixing of the chain is poor because
a large number of steps are required to produce an effectively independent sample. 
Step size tuning must find a practical solution
to this trade-off, and is problem dependent. 
Optimal or practical choices
of the step size 
may depend, among other things, 
on 
the choice of proposal distribution, 
the computational resources available,
the (apparent or effective) dimension of the problem,
and the overall desired accuracy of the MCMC computation.
Lastly, in an $n$-dimensional problem, one can propose an $n$-dimensional update
via an $n$-dimensional proposal, or one can propose, at each step in the chain, an 
update for an $n/m$-dimensional block of variables. 
Such samplers are called within-Gibbs samplers,
partially updating MCMC, component-wise MCMC, or partial resampling algorithms;
see, e.g., \cite{NG06,JJN13,B17}. 

The distributions one wishes to sample by MCMC are often high dimensional. Yet the convergence of MCMC samplers often slows in high dimensions, to the extent that calculations become practically infeasible.
Our main motivation for this work is that, while it is certainly difficult to sample generic high dimensional distributions, distributions that exhibit certain special structure can be feasible to sample, independent of their dimension, provided that the sampler exploits this structure. 
Examples of samplers in the current literature
that leverage various special problem structures are given in Section~\ref{sec:Literature}.
In this paper, we focus on high-dimensional sampling via
the MALA-within-Gibbs sampler in the presence of
sparse conditional structure.

We define sparse conditional structure in Section~\ref{sec:Theory}
via the Hessian of the logarithm of the target density.
In the special case of a Gaussian target distribution, 
sparse conditional structure is equivalent to the precision matrix being sparse.
More generally, 
sparse conditional structure is equivalent to the existence of many conditional independence relationships, or the distribution being Markov with respect to a sparse graph \cite{lauritzen1996graphical}.
We prove in Section~\ref{sec:Theory} that the partial updating strategy of MALA-within-Gibbs, 
with carefully defined updates that make use of the sparse conditional structure, 
leads to acceptance ratios that depend on the dimension of the block-update but are independentof the overall dimension. 
We further show that MALA-within-Gibbs converges with a rate independent of the dimension 
if the target distribution is block-wise log-concave. 

We then discuss MALA-within-Gibbs from a practical perspective in Section~\ref{sec:Practice}.
In this context it is important to realize that
the sparse conditional structure may become apparent only after a suitable change of coordinates.
For MALA-within-Gibbs to be an effective sampler,
we thus need a means of discovering these coordinates,
or, equivalently, identifying sparse conditional structure.
We expect that many Bayesian inference problems,
see, e.g., \cite{ChorinHald2013,RC15,ABN17},
are naturally formulated in coordinates that exhibit sparse conditional structure,
but we also consider an example where a coordinate transformation
is required to reveal sparse conditional structure.
We further discuss the overall computational efficiency
of the MALA-within-Gibbs approach and
raise the issue that the partial updating 
of MALA-within-Gibbs requires
$m$ simulations of the numerical model per sample,
$m$ being the number of blocks.
This implies that there may be an optimal, problem-dependent
choice for the dimensionality of the update
that can lead to significant computational savings.
We explore all of the above issues numerically by applying 
MALA-within-Gibbs to two well known test problems:
a log-Gaussian Cox point process \cite{GC11,MSW98} 
and an elliptic PDE inverse problem
\cite{MTWC15,MarzoukOptimalMaps,Dean2007,Bear,DS11,V15}.
We further compare the computational efficiency
of MALA-within-Gibbs to the efficiencies of other samplers
including MALA, pCN \cite{Cotter13},
and manifold MALA (MMALA) \cite{GC11}.

\section{Notation, assumptions, and background}
We consider probability distributions with density functions $c\pi(\bfx)$,
where $c$ is an unknown normalization constant and $\pi$ is a known function.
We partition the $n$-dimensional vector $\bfx$ into $m$ blocks, $\bfx_1,\ldots,\bfx_m$,
where the subscripts are called block indices.
Note that the blocks $\bfx_j$ are not necessarily consecutive elements of the vector $\bfx$.

\subsection{Notation}
MALA will require gradients of the logarithm of the target density $\pi$,
which we write as $\bfv(\bfx)=\nabla_{\bfx} \log \pi(\bfx)$.
Similarly, we sometimes write derivatives of $\pi$ with respect to the blocks as
$\bfv_j(\bfx)=\nabla_{\bfx_j} \log \pi(\bfx)$.
We write $\nabla^2_{\bfx_i,\bfx_j}$ to denote second derivatives with respect to blocks $i$ and $j$,
i.e., $\nabla^2_{\bfx_i,\bfx_j}\log \pi (\bfx)$ is a matrix of size $\text{dim}(\bfx_i)\times \text{dim}(\bfx_j)$. 
Throughout this paper, we use the Euclidean norm for vectors,
i.e., $\|\bfx\| = \sqrt{\bfx^T\bfx}$,
where superscript $T$ denotes a transpose,
and the $l_2$-operator norm, $\| \bfA \|$, for matrices.
We write $\bfA\preceq\bfB$, where $\bfA$ and $\bfB$ are two $n\times n$ matrices,
when the matrix $\bfA-\bfB$ is negative semi-definite.
We write $\lambda_{\min}(\bfA)$ for the smallest eigenvalue of the matrix $\bfA$.

We write conditional densities of one block,
$\pi(\bfx_j|\bfx_{1},\bfx_2,\dots,\bfx_{j-1},\bfx_{j+1},\dots,\bfx_{m})$,
as $\pi(\bfx_j|\bfx_{\setminus j})$,
i.e, the block index set $\setminus j = \{1,2,\dots,j-1,j+1,\dots,m\}$.
More generally, we write $\mathcal{I}$ for a subset of block indices,
i.e., $\mathcal{I}$ is a subset of $\{1,2,\dots, m\}$.
The cardinality of $\calI$ will be denoted as $|\calI|$. 
We denote the complement of $\mathcal{I}$ by $\calI^c$,
i.e., $\calI^c$  is the subset of $\{1,2,\dots, m\}$
which excludes the block indices in $\mathcal{I}$.

An important concept we will use repeatedly is conditional independence.
Conditional independence means that 
conditioning block $i$ on all but a few other
blocks is irrelevant, which we write as
\begin{equation*}
	\bfx_j \; \perp\!\!\!\perp \; \bfx_{ \mathcal{I}_j^c}\; \vert \; \bfx_{\mathcal{I}_j\setminus \{j\}} \, ,
\end{equation*}
where the index set $\mathcal{I}_j$ depends on $j$ and includes the block index $j$
and where $\calI_j\setminus\{j\}$ is the index set $\calI_j$ with index $j$ removed.
In terms of probability distributions,
conditional independence means that
\begin{equation*}
	\pi(\bfx_j|\bfx_{\backslash j}) = \pi(\bfx_j|\bfx_{\mathcal{I}_j\setminus \{j\}})
\end{equation*}
We assume throughout that $\mathcal{I}_j$ has at most $S\ll m$ elements.

\subsection{Assumptions}
We assume throughout this paper that $\pi(\bfx)$
has continuous second derivatives and that
\begin{enumerate}[(i)]
\item
the dimension, $n$, of $\bfx$ and the number of blocks, $m$, are large;
\item
any block $\bfx_j$ is conditionally independent 
of most other blocks. 
\end{enumerate}
We refer to assumption (ii) as sparse conditional structure.
This terminology is inspired by linear algebra and Gaussian $\pi(\bfx)$---a Gaussian with sparse conditional structure is characterized by a sparse precision matrix.
We make assumption~(ii) mathematically more precise in Section~\ref{sec:Theory:Step}.
For simplicity, we assume that $n/m$ (dimension divided by the number of blocks) is
an integer.

\subsection{Background: MALA and MALA-within-Gibbs}
The MALA sampler with $n$-dimensional updates 
and step size $\tau$
generates a sequence of iterates $\bfx^k$ by repeating the following two steps,
starting from a given $\bfx^0$:
\begin{enumerate}
\item Draw a sample $\bfxtilde^k$ from the MALA proposal by
\begin{equation*}
\bfxtilde^k=\bfx^{k}+\tau \bfv(\bfx^{k})+\sqrt{2\tau} \boldsymbol{\xi}^k,
\end{equation*}
where $\boldsymbol{\xi}^k$ is an independent sample from $\mathcal{N}(\mathbf{0},\bfI_n)$. 
\item Accept this proposal with probability 
\begin{equation*}
\alpha(\bfx^k, \bfxtilde^k)=\min\left\{1, \frac{\pi(\bfxtilde^k)\exp(-\frac1{4\tau}\|\bfx^k-\bfxtilde^k-\tau \bfv(\bfxtilde^k) \|^2)}
{\pi(\bfx^k)\exp(-\frac1{4\tau}\|\bfxtilde^k-\bfx^k-\tau \bfv(\bfx^k)\|^2)}\right\},
\end{equation*}
i.e., let $\bfx^{k+1}=\bfxtilde^k$ with probability $\alpha(\bfx^k, \bfxtilde^k)$, 
 and $\bfx^{k+1}=\bfx^k$ with probability $1-\alpha(\bfx^k, \bfxtilde^k)$. 
\end{enumerate}
It is straightforward to show that $c\pi(\bfx)$ is the invariant distribution of the Markov chain $\bfx^k$. Therefore, when $k\to \infty$, $\bfx^k$ can be viewed as a sample from $c\pi(\bfx)$.

MALA-within-Gibbs is a variation of MALA that 
uses $n/m$-dimensional updates.
We use superscripts to index ``time'' in the Markov chain (see above),
and subscripts to index the blocks.
Thus, starting with a vector $\bfx^0$ and step size $\tau$,
MALA-within-Gibbs iterates the following steps: 
\begin{enumerate}
\item Set $\bfx^{k}=\bfx^{k-1}$. Repeat steps (a) and (b) below for $j=1,\ldots, m$ 
to update all $m$ blocks of~$\bfx^{k}=[\bfx_1^k,\dots,\bfx_m^k]$.
\begin{enumerate}[(a)]
\item Sample a standard Gaussian  $\boldsymbol{\xi}_{j}^k$ of the same dimension as $\bfx_j$
and use the MALA proposal for the current block $\bfx_j$:
\begin{equation}
\label{eqn:MALApropose}
\tilde{\bfx}^k_{j}=\bfx^k_{j}+\tau  \bfv_j (\bfx^k)+\sqrt{2\tau } \boldsymbol{\xi}_{j}^k,
\end{equation}
\item Define $\bfxtilde^k=[\bfx_1^k,\dots,\bfx_{j-1}^k,\tilde{\bfx}_j^k,\bfx_{j+1}^k,\dots,\bfx_m^k]$,
i.e., $\bfxtilde^k$ is equal to $\bfx^k$,
except at its $j$-th block.  Compute the block acceptance ratio
\begin{equation*}
\alpha_j(\bfx^k, \bfxtilde^k)=\min\left\{1, \frac{\pi(\bfxtilde^k)\exp(-\frac1{4\tau}\|\bfx^k-\bfxtilde^k-\tau \bfv_j(\bfxtilde^k) \|^2)}
{\pi(\bfx^k)\exp(-\frac1{4\tau}\|\bfxtilde^k-\bfx^k-\tau \bfv_j(\bfx^k)\|^2)}\right\},
\end{equation*}
Set $\bfx^k$ be $\bfxtilde^k$ with probability $\alpha_j(\bfx^k,\bfxtilde^k)$, else $\bfx^k$ maintains its value.
\end{enumerate}
\item  Increase the time index from $k$ to $k+1$ and go to 1.
\end{enumerate}
As before, it is straightforward to verify that the target $c\pi(\bfx)$ is the invariant distribution
of the MALA-within-Gibbs iterates.
The partial updating can be derived from
applying MALA within a Gibbs iteration (hence the name),
i.e., with target distributions $\pi(\bfx_j|\bfx_{\backslash j})$.

\section{Dimension independent acceptance and convergence rate}
\label{sec:Theory}
Both the MALA and MALA-within-Gibbs samplers can, in principle,
be used for arbitrary target distributions,
but in generic high-dimensional problems we expect that convergence is slow.
In high-dimensional problems with sparse conditional structure,
however, MALA-within-Gibbs can be effective if the partitioning of $\bfx$,
defining the partial updates, is chosen in accordance with the sparse conditional structure.
With a suitable partial updating strategy,
we show that the step size and the acceptance ratio 
of MALA-within-Gibbs (within each block) can be made independent of the overall dimension.
We then show, under additional assumptions of block-wise log-concavity,
that MALA-within-Gibbs converges to the target distribution at a dimension-independent rate.
The proofs of the propositions and the theorem can be found in the appendix.

\subsection{Dimension independent acceptance rate under sparse conditional structure}
\label{sec:Theory:Step}
To simplify the proofs,
the conditional independence assumption is formulated in terms of the gradient $\bfv(\bfx)$.
\begin{aspt}[Sparse conditional structure]
\label{aspt:sparse}
For $\pi(\bfx)$ and the partition $\bfx=(\bfx_1,\ldots, \bfx_m)$, 
there are  constants $S$ and $q$ independent of $n$, so that
\begin{enumerate}[(i)]
\item  The dimension of each block $\bfx_j$ is bounded by $q$.
\item For each block index $j\in {1,\ldots, m}$, there is an $\bfx$-independent block index set $\calI_j\subset\{1,\ldots,m\}$ with $j\in \calI_j$ and cardinality $|\calI_j|\leq S$ 
so that 
\[
\nabla^2_{\bfx_k,\bfx_j} \log \pi(\bfx)=\nabla_{\bfx_k} \bfv_j(\bfx)=\mathbf{0},\quad \text{if}\quad k\notin \calI_j. 
\]
\end{enumerate}
\end{aspt}
Note that (i) is trivial because we deal with finite dimensional problems
and note that (ii), by Lemma~2 of \cite{SBM18}, is equivalent to
$\bfx_j\perp\!\!\!\perp \bfx_{\calI_j^c} \; | \; \bfx_{\calI_j\backslash \{j\}}$ 
if the density is strictly positive and smooth.
In other words, Assumption \ref{aspt:sparse} 
is equivalent to the assumption of sparse conditional structure, as described earlier,
but the formulation in terms of gradients is easier to use in our proofs.

Whether or not Assumption~\ref{aspt:sparse} is satisfied
for a given target distribution depends, to a large extent, 
on how the blocks of $\bfx$ are defined.
Using physical insight into the problem,
it is often possible to group components of $\bfx$
such that Assumption~\ref{aspt:sparse} is satisfied
or approximately satisfied.
We discuss this issue more in Section~\ref{sec:Practice} below,
but it is important to understand that Assumption~\ref{aspt:sparse}
essentially requires a good understanding of the target distribution
and that the results we derive under this assumption
make use of the fact that one understands and leverages
conditional independencies among the components of $\bfx$.

We also assume that the gradient $\bfv(\bfx)$ and its derivatives are bounded.
\begin{aspt}[Bounded vector fields]
\label{aspt:regular}
The vector field $\bfv_j(\bfx)=\nabla_{\bfx_j} \log \pi(\bfx)$,
for $j=1,\dots,m$, and its first derivatives  are bounded, 
i.e., there exist constants $M_v$ and $H_v$, 
independent of the overall dimension $n$, such that
\[
\|\bfv_j(\bfx)\|\leq M_v,\quad \|\nabla_{\bfx_i} \bfv_j (\bfx)\|\leq H_v,\quad \|\nabla_{\bfx_j} \bfv_j (\bfx)-\nabla_{\bfz_j} \bfv_j (\bfz)\|\leq H_v\|\bfx-\bfz\|. 
\]
\end{aspt}

By Assumption~\ref{aspt:sparse}, $\bfv_j(\bfx)$ has no dependence on $\bfx_{\calI^c_j}$, 
so one can write it as $\bfv_j(\bfx_{\calI_j})$. 
If the support of $\pi(\bfx)$ is bounded,
then $\bfv_j(\bfx)$ having no dependence on $\bfx_{\calI^c_j}$,
along with the fact that the dimension of each block $\bfx_i$ is at most $q$, 
often yields Assumption~\ref{aspt:regular}.
Unbounded support is more complicated.
A Gaussian, for example, violates Assumption~\ref{aspt:regular}
because the norm of the gradient is not bounded.
This boundedness assumption, however, 
is made for simplicity and may not be required in practice.
More sophisticated constructions may be used in the future
to relax this assumption and to derive more general results.

Under Assumptions~\ref{aspt:sparse} and~\ref{aspt:regular},
the  following proposition shows that the step size and 
the acceptance ratio of MALA-within-Gibbs are independent of the overall dimension. 
\begin{prop}[Block acceptance]
\label{prop:BlockAcceptance}
Suppose $c\pi(\bfx)$ is the density of a distribution with sparse conditional structure (Assumption \ref{aspt:sparse})
and that, in addition, Assumption~\ref{aspt:regular} holds.
There is a constant $M$, independent of the number of blocks $m$, 
so that, for any given state $\bfx\in \reals^{n}$, 
the block acceptance ratio $\alpha_j(\bfx^k,\bfxtilde^k)$ is bounded below by
\[
\E[ \alpha_j(\bfx^k,\bfxtilde^k)]\geq 1-M\sqrt{\tau }.
\]
for all blocks $j\in\{1,\ldots, m\}$.
\end{prop}
Proposition~\ref{prop:BlockAcceptance}
follows directly from Lemma \ref{lem:couplingprob},
which we prove in Section~\ref{sec:proof}.
The dimension independent block acceptance ratio is intuitive.
Partial updating implies that the proposed updates are, by design, low-dimensional:
their dimension depends on the block size, $q$, 
but is independent of the number of blocks, $m$,
or the overall dimension $n=m\cdot q$.
Thus, only the dimension of the update controls the block acceptance ratio.
The overall number of low-dimensional updates,
which defines the overall dimension, is irrelevant.

It might seem that Proposition~\ref{prop:BlockAcceptance}
contradicts earlier results on optimal scalings of MALA step sizes,
where the optimal step size decreases with dimension at a well-understood rate
\cite{Roberts97,Roberts98,Beskos09,Beskos13}.
These earlier scaling results, however,
do not assume sparse conditional structure.
Thus, in general, the optimal step size of MALA and MALA-within-Gibbs
should decrease with dimension,
but if the target distribution has sparse conditional structure
and if, in addition, this structure is used in the block updates,
then the step size (and acceptance ratio) can be independent of the overall dimension.

Assumption~\ref{aspt:sparse} ensures
that the sparse conditional structure of the target is exploited by the MALA-within-Gibbs sampler.
We thus assume away any difficulties of discovering sparse conditional structure,
but we discuss practical aspects of this assumption in Section~\ref{sec:Practice},
including a brief discussion of what happens when the assumptions are only
nearly met.
We also emphasize that we have no claims at optimal step sizes of MALA-within-Gibbs---we merely show that the step size need not decrease with dimension
to ensure a constant average block acceptance ratio. Moreover,
local tunings as discussed in \cite{B17}
may further improve efficiency, but we do not pursue such ideas here.

Partial updating of MALA has also been considered in \cite{NG06},
where the conclusion is that the updates in MALA-within-Gibbs should be high-dimensional.
Again, this is true in general, but if the target has sparse conditional structure,
the dimensionality of the updates may depend on this structure.
We revisit this issue in Section~\ref{sec:Practice},
where we also bring up a trade-off between the block size
and computational requirements that may increase with the number of blocks.
One can also perform random partial updates,
i.e., choosing at random which components of $\bfx$ are next updated.
Asymptotically, MALA-within-Gibbs
with a random partial updating strategy converges to the target distribution,
but we expect that the convergence will be slow 
for problems with sparse conditional structure
because this structure is not used by random partial updates.

\subsection{Dimension independent convergence rate}
\label{sec:Theory:Conv}
We have shown above that the step size and acceptance ratio
of MALA-within-Gibbs can be independent of dimension
if the sparse conditional structure of the target distribution is known 
and used via a suitable partition of the variables during the within-Gibbs moves.
This is not enough to guarantee fast convergence of MALA-within-Gibbs.
To study the convergence rate of MALA-within-Gibbs,
we require, as an additional assumption,
that the target distribution be unimodal and block-wise log-concave
(see below for a definition).
The reason is that difficulties with MCMC that arise from high dimensionality or
multi-modality are independent of each other:
if the target distribution has multiple modes,
a large number of samples may be required even if the dimension is small.
We focus on aspects of high-dimensional problems with a single mode.

The additional assumption we need in our proof (see Section~\ref{sec:Proofs})
is block-wise log-concavity.
To define block-wise log-concavity,
we first construct an $m\times m$ matrix $\bfH(\bfx)$,
where $m$ is the number of blocks,
with the following properties.
\begin{defn}
A symmetric $m\times m$ matrix function $\bfH(\bfx)$  with entries $H_{j,i}(x)$ is uniformly bounded and negative if there are strictly positive constants $H_v$ and $\lambda_H$ such that for all $j,i$ and all $\bfx$,
\[
|H_{j,i}(\bfx)|\leq H_v,\quad \lambda_{\max}(\bfH(\bfx))\leq -\lambda_H<0. 
\]
\end{defn}
As a simple example, a constant symmetric negative definite matrix is uniformly bounded and negative. Block-wise log-concavity can now be formulated as follows.
\begin{aspt}
\label{aspt:blockconcave}
A probability density $c\pi(\bfx)$ is block-wise log-concave with block size $m$
if there exists a $m\times m$ uniformly bounded and negative matrix $\bfH(\bfx)$ such that
\begin{enumerate}[i)]
\item $\nabla^2_{\bfx_j,\bfx_j} \log \pi(\bfx)\preceq H_{j,j}(\bfx)\, \bfI$, where $\bfI$ is 
the identity matrix of size $\text{dim}(\bfx_j)\times \text{dim}(\bfx_j)$;
\item the off-diagonal elements bound the conditional dependence between blocks, 
that is, $\|\nabla^2_{\bfx_j,\bfx_i}\log \pi(\bfx)\|\leq H_{j,i}(\bfx)$ for all $i\neq j$.
\end{enumerate}
\end{aspt}

Note that if the dimension of the blocks is $q=1$,
so that $m=n$, and if the Hessian of $\log \pi(\bfx)$ is diagonally dominant,  
then $\bfH(\bfx)$ can be taken as the Hessian of $\pi$, 
with all off-diagonal entries replaced by their absolute value. 
Further note that block-wise log-concavity is a stronger assumption
than log-concavity---the function $\pi(\bfx)$ can be log-concave but not block-wise log-concave
(see example below).
On the other hand, a distribution that is block-wise log-concave, for any block size, is also log-concave.

As an illustration, we consider a Gaussian distribution for 
$\bfx = [x_1,\dots,x_{64}]$ with mean zero and covariance matrix $\bfC$ with elements
\[
 [C]_{i,j}=\exp\left(-\frac1{2l} |i-j|\right),\quad i,j=1,\ldots, 64. 
\]
Interpreting this Gaussian as a discretization of a 1D random field
with exponential covariance kernel (and discretization $\Delta x =1$),
the quantity $l$ is a correlation length scale.
If $l$ is small, only those components of $\bfx$ that are near each other
in the 1D domain are significantly correlated.
This suggests partitioning $\bfx$ based on neighborhoods in the 1D domain,
which correspond to consecutive elements of $\bfx$.
For example, the block size $q=4$ results in $m=16$ blocks
\[
\bfx_1=[x_1,\ldots,x_4],\quad \bfx_2=[x_5,\ldots,x_8],\quad \ldots,\quad\bfx_{16}=[x_{61},\ldots,x_{64}]. 
\]
Recall that, for Gaussian distributions, the precision matrix $\bfP$ is equal to $-2\nabla^2 \log \pi$,
which suggests to construct the matrix $\bfH(x)\in \reals^{m\times m}$ by 
\[
H_{i,i}(\bfx)\equiv -\lambda_{\min}(\bfP_{i,i}),\quad H_{i,j}(\bfx)\equiv \|\bfP_{i,j}\|.
\]
Here, $\bfP_{i,j}$ is the $i,j$-th $q\times q$ sub-block of $\bfP$ with indices corresponding to the blocks $\bfx_i$ and $\bfx_j$. 
For example, with $q=4$, $\bfP_{1,2}$ is a sub-block of $\bfP$ 
consisting of rows 1-4 and columns 5-8. 
Assumption \ref{aspt:blockconcave} is then equivalent to assuming that $\bfH$ is negative definite, 
i.e., $\lambda_{\min}(-\bfH)>0$. We can numerically check this condition
by computing eigenvalues of $\bfH$.
Table~\ref{tab:concave} lists values of $\lambda_{\min}(-\bfH)$ 
for varying correlation length scales $l$ and block sizes $q$. 

\begin{table}[tb]
\footnotesize
\begin{center}
\begin{tabular}{ c c c c c c c}
 Setting &q=1 &q=2 & q=4& q=8 &q=16 &q=32\\\hline
$l=2$ &-0.71 &-1.28 &-1.44 &-1.28 &-0.71 &\textbf{0.24}\\
$l=1$ &\textbf{0.04} &-0.21 & -0.29 & -0.21 & \textbf{0.04} &\textbf{0.46}\\
$l=0.5$ &\textbf{0.62} &\textbf{0.54} &\textbf{0.52} & \textbf{0.54} & \textbf{0.62} & \textbf{0.76}
\end{tabular}
\end{center}
\caption{Block-wise log-concavity with $\bfH$ as defined in the text for different block sizes $q$ and different correlation length scales $l$.
Positive numbers, highlighted in bold, indicate block-wise log-concavity.}
\label{tab:concave}
\end{table}

We note that while the Gaussian is log-concave for any $l$,
block-wise log-concavity depends on the length scale $l$ and the 
size of the blocks $q$.
If $l$ is large, only large blocks lead to block-wise log concavity
(with $q=64$ guaranteeing log-concavity and block-wise log-concavity).
If the correlation is (essentially) confined to small neighborhoods, i.e., if $l$ is small,
then small block sizes $q$ also lead to block-wise log-concavity.

With the definition of block-wise log-concavity,
we can now state a theorem about the dimension-independent convergence
rate of MALA-within-Gibbs. The proof is given in Section~\ref{sec:Proofs}.
\begin{thm}
\label{thm:unbiased}
Under Assumptions \ref{aspt:sparse}, \ref{aspt:regular}, and \ref{aspt:blockconcave}, for any $\delta>0$, there exists a $\tau_0>0$ independent of the number of blocks $m$, so that when the step size $\tau<\tau_0$, we can couple two MALA-within-Gibbs samples $\bfx^k$ and $\bfz^k$, such that 
\[
\sum_{i=1}^m\left(\E \| \bfx^k_i-\bfz^k_i\|\right)^2\leq  (1-(1-\delta)\lambda_H \tau)^{2k} \sum_{i=1}^m\left(\E \| \bfx^0_i-\bfz^0_i\|\right)^2. 
\]
In particular, one can let $\bfz^0\sim \pi$.
It follows  that $\bfz^k\sim \pi$, 
which in turn shows that $\bfx^k$ converges to $\pi$ geometrically fast. 
\end{thm}

Theorem~\ref{thm:unbiased} indicates that MALA-within-Gibbs can be 
a fast sampler for high-dimensional problems if
(\emph{i}) the target distribution has sparse conditional structure
and this structure is used by the MALA-within-Gibbs sampler;
and (\emph{ii}) the target distribution is block-wise log-concave.

Block-wise log-concavity implies that the target has only one mode.
Assuming block-wise log-concavity thus allows us to study computational barriers due to high dimensionality
without requiring that we simultaneously consider challenges due to multi-modality.
Notably, block-wise log-concavity is more restrictive than log-concavity,
which also implies that the target is unimodal.
We use block-wise log-concavity here to gain stronger control 
over the coupling between blocks in the analysis.
Ultimately, a less restrictive assumption (e.g., log-concavity) may be preferable
and one may view our results as a first step towards a full understanding
of how MALA-within-Gibbs can operate effectively in high-dimensional problems.

Theorem~\ref{thm:unbiased}
also has connections to previous work
on Gibbs samplers for Gaussian distributions with sparse conditional structure \cite{MTM19}.
In particular, Theorems 3.1 and 3.2 of \cite{MTM19} show dimension independent
convergence of a Gibbs sampler for Gaussian distributions.
By interpreting the block-wise log-concavity assumption as a generalization of the 
Gaussian assumptions in Theorems 3.1 and 3.2 of \cite{MTM19},
one can understand Theorem~\ref{thm:unbiased} as a generalization  
of this result to a within-Gibbs sampler for non-Gaussian distributions.

\section{Discussion of efficient samplers in high dimensions}
\label{sec:Literature}
We suggested earlier that sampling generic high-dimensional distributions is difficult,
but if the target distribution has a special structure, then efficient samplers can be constructed.
One example are Gaussian distributions.
Gaussians can be sampled efficiently even if their dimension is large, 
either by direct samplers (using techniques from numerical linear algebra for computing matrix square roots),
or by MCMC, using analogies between Gibbs samplers and linear solvers to construct matrix splittings for accelerated sampling \cite{Fox,FN16}. 
Connections between linear solvers and MCMC are also discussed in \cite{MJ95}.

There are also several routes to making MCMC samplers effective for high-dimensional distributions 
that are not Gaussian, and we discuss some of them here in relation to our results.
Recall that the MALA-within-Gibbs sampler relies on sampling blocks of variables; this idea is also used in multigrid Monte Carlo (MGMC) 
for lattice systems \cite{MGMC86,MGMC89,MGMC91}.
The basic ideas in MGMC, however, are different.
MGMC proposes the same move simultaneously for all variables within one block.
The MALA-within-Gibbs sampler proposes moves for each variable within a block,
but handles the blocks (nearly) independently.
In terms of analogues to linear solvers,
MALA-within-Gibbs is perhaps more akin to domain decomposition than to multigrid.

Metropolis-Hastings (MH) samplers offer a general tool for sampling non-Gaussian target distributions.
The step size $\tau$ of many MH samplers needs to be tuned,
which often requires that $\tau$ decrease with dimension $n$.
In fact optimal scalings of $\tau$ with dimension for various MCMC samplers
have been derived:
for RWM $\tau_\text{opt}=O(n^{-1})$, 
for MALA $\tau_\text{opt}=O(n^{-1/3})$, 
and for HMC $\tau_\text{opt}=O(n^{-1/4})$. 
Fixing the acceptance ratio with small $\tau$, however, comes with a price:
the acceptance ratio may be large, 
but all accepted steps are small (on average on the order of $\tau$),
so that the sampler moves often, but slowly.
As a rule of thumb, it takes about $O(1/\tau)$ iterations to move through the 
support of the target distribution (putting aside issues of reaching stationarity 
\cite{christensen2005scaling}). 
For very large dimensions, many MCMC samplers are thus slow to converge.
These results hold for general target distributions.
Even if the analyses that lead to the optimal scalings rely on certain assumptions, 
e.g., that the target measure is of product type,
this problem structure is not directly used by the samplers.
Our results on dimension independent step size and convergence rates
hold only for target distributions with sparse conditional structure
and when this structure is explicitly used by the MALA-within-Gibbs sampler
(via Assumption~\ref{aspt:sparse}).

Another strategy for effective sampling of high-dimensional distributions applies if the parameters $\bfx$ can be decomposed as $\bfx=(\bfy, \bfz)$, where $\bfz$ is low dimensional and, conditioned on $\bfz$, there are fast (perhaps direct) samplers for $\bfy$; see, e.g., \cite{CM16,CM18, CMT18UQ}.  In Bayesian inverse problems, this structure can often be identified by a suitable choice of basis or reparameterization, as in \cite{CuiEtAl16,SpantiniEtAl15,zahm2018certified}: typically $\bfz$ represents directions where the posterior departs significantly from the prior, while $\bfy$ represents prior-dominated directions of the parameter space, which may even be conditionally Gaussian and/or (approximately) independent of $\bfz$.  
Note that the MALA-within-Gibbs sampler does not require partially or conditionally Gaussian target distributions, but our analysis of its efficiency requires sparse conditional structure and block-wise log-concavity.

The theory of function-space MCMC also has led to effective MCMC methods for another class of high-dimensional Bayesian inverse problems; see, e.g., \cite{Cotter13,hairer2014spectral,V15, OPPS16}. 
The basic idea is to design MCMC samplers that are well defined on
function spaces: consider, for example, the inference of a spatially
distributed parameter, and send the spatial discretization parameter $h$ to zero
while keeping the number of observations fixed.
Since the spatial discretization controls the dimension of the problem,
the limit $h \to 0$ corresponds to an infinite-dimensional problem.
The proposal distributions of function-space MCMC
are chosen such that, for a fixed step size, the acceptance ratio remains constant,
independent of $h$ or, equivalently, the dimension of the problem.
There are many variations of such discretization invariant MCMC samplers \cite{CuiEtAl16, rudolf2018generalization,chen2016accelerated,beskos2014stable},
with applications discussed in, e.g., \cite{BuiEtAl13,PetraEtAl14}. Some (e.g., \cite{CuiEtAl16b,bardsley2019scalable,lan2019adaptive}) combine function-space MCMC with the decomposition of the parameter space into (finite dimensional) likelihood-informed directions and (infinite-dimensional) prior-dominated directions, as described above.

The dimension independence of the MALA-within-Gibbs sampler discussed here is markedly different from the discretization invariance of function-space MCMC. We do not consider the infinite-dimensional limit resulting from the discretization of a function on a given domain. Rather, we show that the dimension of the inference problem may not affect the acceptance and convergence rates of the sampler if the problem has sparse conditional structure that is suitably exploited.
This allows us, for example, to consider a sequence of inference problems posed on increasingly large domains and with an increasing number of observations,
while keeping the spatial discretization fixed (see Section~\ref{sec:LogGauss} and Figure~\ref{fig:Truth}). Provided that each of the problems within the sequence has suitable sparse conditional structure, the acceptance and convergence rates of MALA-within-Gibbs are the same for all problems within the sequence.
The convergence rate of function-space MCMC, e.g., the pCN method, is not constant for all problems within such a sequence (see Section~\ref{sec:LogGauss} for a specific example).
This does not come as a surprise, because the dimension independence of pCN  holds only in the case in which the dimension increases due to grid refinement, with the number of observations and the domain size held constant. A sequence of problems with increasing domain size and increasing numbers of observations violates the assumptions underpinning the dimension independence of pCN.
%
%
We refer to \cite{MTM19} for a more thorough discussion of the various notions of high dimensionality that may occur when solving Bayesian inverse problems.
%
%

\section{Practical considerations and numerical experiments}
\label{sec:Practice}
Our results on dimension independent step size, acceptance probabilities, 
and convergence rates hold under precise mathematical assumptions
of sparse conditional structure and log-concavity (see Section~\ref{sec:Theory}).
We now focus on posterior distributions that arise in Bayesian inference problems, 
because of their practical importance and because we anticipate
that the assumption of sparse conditional structure is often satisfied in such problems.
We also demonstrate how to use MALA-within-Gibbs in two numerical examples,
and discuss and compare the computational costs of MALA-within-Gibbs and other MCMC samplers.

\subsection{Posterior distributions with sparse conditional structure}
In below, we provide the Bayesian problem setup.
Let $\bfx$ be an $n$-dimensional vector endowed with a prior probability density $\pi_0(\bfx)$. 
In many problems, $\bfx$ arises from a discretization of a spatially distributed quantity (i.e., a field) and, for that reason, is high dimensional.
The prior reflects assumptions about the smoothness of the field 
and is often assumed to be Gaussian with a known mean and covariance.
A computational model, $\mathcal{M}(\bfx)$,
maps $\bfx$ to observations $\bfy$. 
Typically, the model is nonlinear and the number of observations 
is less than the dimension of $\bfx$.
Any model errors are represented by a random variable $\boldsymbol{\varepsilon}$
and, often, model errors are additive, i.e.,
\begin{equation}
\label{eq:Likelihood}
	\bfy =\mathcal{M}(\bfx)+\boldsymbol{\varepsilon}.
\end{equation}
The distribution of $\boldsymbol{\varepsilon}$ is assumed to be known 
(often Gaussian with mean zero and diagonal covariance matrix).
Equation~\eqref{eq:Likelihood} defines a likelihood $\pi_l(\bfy\vert\bfx)$
and the likelihood and prior jointly define the posterior distribution
\begin{equation*}
	\pi(\bfx\vert\bfy) \propto \pi_0(\bfx) \pi_l(\bfy\vert\bfx).
\end{equation*}

Sparse conditional structure arises naturally in Bayesian posterior distributions when
(\emph{i}) the parameters $\bfx$ are high dimensional,
but not all components of $\bfx$ have significant statistical interactions;
and (\emph{ii}) each observation is informative for only a small subset of the components of $\bfx$
(see also \cite{MTM19}).
Put differently, we assume that the prior has sparse conditional structure
and that the observations do not significantly densify the conditional structure of the prior.
This happens in many geophysical applications, e.g., 
in numerical weather prediction (NWP),
where the posterior distribution is defined jointly by a global atmospheric model 
(with dimension $O(10^8)$) and observations of the atmospheric state 
(typically $O(10^7)$ observations).
In a global atmospheric model,
each model component stores information about the atmospheric state at a specific location at a given time
and each component has significant statistical interactions with nearby components,
but not with components that are far away.
A discussion of the mathematical mechanisms 
that lead to this property can be found in \cite{CMT18CAM}.

\subsection{Implementation of MALA-within-Gibbs}
The partitioning of $\bfx$ into blocks is important
for effective sampling of the posterior by MALA-within-Gibbs,
because only a suitable partition will indeed put the sparse conditional structure to use. 
We do not have a general strategy to find a suitable partition,
but we expect that a workable partition is often intuitive.
For example, if $\bfx$ is defined over a spatial domain (1D, 2D, or 3D)
and if correlations are limited to small neighborhoods,
then the partitioning should be based on these neighborhoods
and the block size should take the correlation lengths scales into account.
We demonstrate this process in a numerical example in Section~\ref{sec:LogGauss}.
Our second example in Section~\ref{sec:Subsurface}
demonstrates how to choose a partition for the partial updating
based on prior covariances in the absence of a spatial scale.

The partial updating of MALA-within-Gibbs
requires $m$ likelihood evaluations per sample.
In a typical Bayesian inverse problem,
each likelihood evaluation will require a full forward solve with the numerical model $\mathcal{M}$,
even if only one block of the model's components is updated.
Using the common effective sample size
\begin{equation}
	N_\text{eff} = N_e/\text{IACT},
\end{equation}
where $N_e$ is the number of MCMC samples 
(length of the Markov chain) and IACT is the 
integrated auto correlation time, see, e.g., \cite{Wolff04,Sokal1998},
we can estimate the cost per effective sample by
\begin{equation}
\label{eq:CostPerSample}
\text{cost per effective sample} = \text{IACT} \times (\text{\# of blocks}) \times \text{cost of likelihood evaluation}.
\end{equation}
It is now clear that the computational cost of MALA-within-Gibbs grows
with the number of blocks,
even if IACT is independent of dimension.
The computational requirements of MALA-within-Gibbs
are, therefore, not independent of the dimension of the problem,
since a higher-dimensional problem requires a larger number of blocks
(keeping other parameters that define the model unchanged; see examples below).
This also points to a trade-off for sampling in high dimensions
that may not be easy to resolve:
to keep the efficiency high (small IACT), one may want to use a large number of small blocks
(with a lower bound on the block size depending on the correlation structure),
but on the other hand, one may want to use small number of large blocks to keep the number of model evaluations per sample small.

\subsection{Numerical illustration 1: Log-Gaussian Cox point processes}
\label{sec:LogGauss}
We consider inference in a log-Gaussian Cox point process
similar to the numerical experiments in \cite{GC11}.
A uniform $N_u\times N_v$ grid, with spacing $\Delta u = \Delta v= 1$, covers
the 2D (spatial) domain, $[1,L]\times[1,L]$. 
The parameter to be inferred is defined over the domain, $\{X_{i,j}, i,j=1,\dots,L/\Delta u\}$. Its prior is $\mathcal{N}(\mu\bfOne,\bfB)$,
where $\bfB$ is a discretization of the exponential covariance kernel, i.e.,
\begin{equation*}
	\text{cov}(X_{s_1,t_1}, X_{s_2,t_2}) = \sigma_s^2 \sigma_t^2
	\exp\left(-\frac{1}{2}\frac{\vert s_1-s_2\vert}{l_s}-\frac{1}{2}\frac{\vert t_1-t_2\vert}{l_t}\right),
\end{equation*}
 where $\sigma_s^2=\sigma_t^2=2$, $\mu = 4$, $l_s=2$, $l_t=4$.
 
 Observations are made at each grid point, denoted by $Y_{i,j}$. The observations are conditionally independent and Poisson distributed with means $\exp(X_{i,j})$. Our goal is to estimate $X_{i,j}$ from $Y_{i,j}$. The prior and likelihood define the posterior distribution
\begin{equation*}
	\pi(\bfx\vert  \bfy)\propto
	\exp\left(-\frac{1}{2} \vert \vert\bfB^{-1/2}(\bfx-\mu\bf1) \vert \vert^2\right)
	 \prod_{i,j} \exp\left(Y_{i,j}X_{i,j}-\exp(X_{i,j})\right),
\end{equation*}
where $\bfx$ is the column stack of $X_{i,j}$,
i.e., an $n=L^2$ dimensional vector.

\subsubsection{Problem setup}
We consider three problems
with increasing domain size $L=16$, $L=32$, and $L=64$,
leading to sampling problems of dimensions $256$, $1024$, and $4096$.
The true values of $X_{i,j}$ for the three domains are shown in Figure~\ref{fig:Truth}.
\begin{figure}[tb]
\begin{center}
\includegraphics[width=0.8\textwidth]{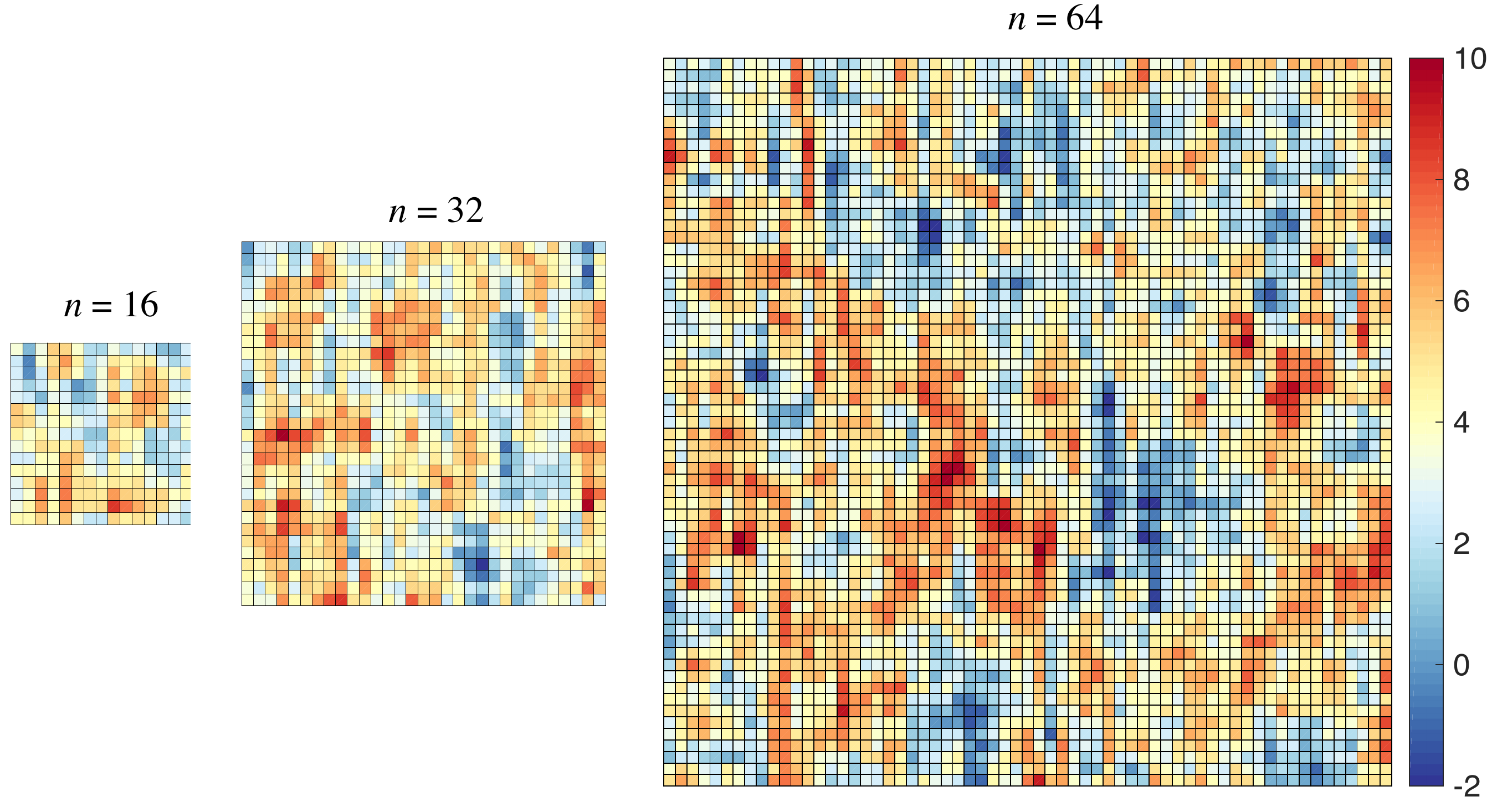}
\caption{True values of $X$ for three problems with increasing domain size $L$
(drawn to scale).}
\label{fig:Truth}
\end{center}
\end{figure}
Note that the (apparent) dimension of the problem ($n=L^2$)
and the number of observations ($L^2$) increase with increasing domain size,
but the prior length scales are fixed and short compared to all three domain sizes.
Moreover, each observation $Y_{i,j}$ carries information about only one grid point, $X_{i,j}$.

We note that our setup is different from the problems usually considered in function-space MCMC,
where the increasing dimension is caused by refining the spatial discretization,
while keeping the size of the domain and the number of observations constant.
In fact, we consider the opposite scenario: the number of observations and the size of the domain increase,
but the spatial discretization remains unchanged.
Our setup is also slightly different from that considered in \cite{GC11},
where the means of the Poisson distributions are $\exp([X]_{i,j})/L^2$ (in our notation)
and where a different covariance kernel is used to define the prior.
The latter is minor.
We do not scale the mean values with domain size, $L$,
because we want to study MALA-within-Gibbs on problems with increasing
dimension while leaving all other parameters that define the problem unchanged.

The prior precision matrix is sparse,
as illustrated in Figure~\ref{fig:PriorPrecision}
for the problem of size $16\times 16$.
\begin{figure}[tb]
\begin{center}
\includegraphics[width=0.4\textwidth]{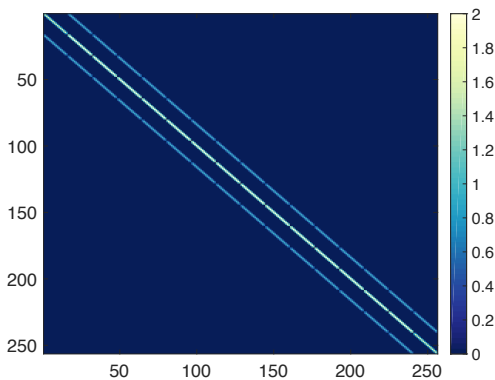}
\caption{Prior precision matrix of the $16\times 16$ problem.}
\label{fig:PriorPrecision}
\end{center}
\end{figure}
The prior precision matrices of the larger problems ($32\times 32$ and $64\times 64$)
have similar sparsity patterns.
We now investigate whether Assumptions \ref{aspt:sparse} (sparse conditional structure) and \ref{aspt:blockconcave} (block-wise log-concavity) are satisfied in this problem. Sparse conditional structure can verified by inspection: the chosen Gaussian prior is a Markov random field where each pixel has only four neighbors, and the likelihood is purely local, introducing no new dependencies. We can verify this structure more carefully as follows, partitioning the state $\bfx$ based on 2D-neighborhoods. Recall that the log posterior density is
\begin{equation}
\label{eq:CoxLogPost}
\log \pi(\bfx\vert  \bfy)=C-\frac{1}{2} \vert \vert\bfB^{-1/2}(\bfx-\mu\mathbf{1}) \vert \vert^2
	 +\sum_{i,j} \left(Y_{i,j}X_{i,j}-\exp(X_{i,j})\right),
\end{equation}
where $C$ is a constant whose value is irrelevant.
Fixing the block size at $q=n/m$, we find that
\begin{equation}
\label{temp:hess}
\nabla_{\bfx_i,\bfx_j}\log \pi(\bfx\vert  \bfy)= -[\bfB^{-1}]_{\bfx_i,\bfx_j}-\mathbf{1}_{i=j}\bfD_i,
\end{equation}
where the $q\times q$ matrices $[\bfB^{-1}]_{\bfx_i,\bfx_j}$ are  
constructed from $\bfB^{-1}$ based on the blocks $\bfx_i,\bfx_j$,
and $\bfD_i$ is a diagonal $q\times q$ matrix with entries being 
$\exp(X_{k,l})$ for each $X_{k,l}$ in $\bfx_i$. 
Due to the sparse structure of the prior precision $\bfB^{-1}$ 
(see Figure~\ref{fig:PriorPrecision}),
$[\bfB^{-1}]_{\bfx_i,\bfx_j}$ is zero if 
the blocks $i$ and $j$ are far from each other in the 2D domain.
In this case, $\mathbf{1}_{i=j}\bfD_i=0$, so that 
by~\eqref{temp:hess},  $\nabla^2_{\bfx_i,\bfx_j}\log \pi(\bfx\vert  \bfy)$  is also zero.
The problem is thus indeed characterized by sparse conditional structure.

The block-wise log-concavity Assumption~\ref{aspt:blockconcave} 
may not be satisfied in this example.
By \eqref{temp:hess}, 
the Hessian is bounded above by $-\bfB^{-1}$,
which suggests to use 
\begin{equation*}
 	H_{i,i}(x)\equiv-\lambda_{\min} ([\bfB^{-1}]_{\bfx_i,\bfx_i}),\quad H_{i,j}(x)\equiv\|[\bfB^{-1}]_{\bfx_i,\bfx_j}\|,\quad  i,j=1,\ldots,m
\end{equation*}
where $[\bfB^{-1}]_{\bfx_i,\bfx_i}$ and $[\bfB^{-1}]_{\bfx_i,\bfx_j}$
are constructed from $\bfB^{-1}$, based on the blocks with indices $i$ and $j$.
With this choice, Assumption~\ref{aspt:blockconcave} requires that 
\begin{equation}
	\label{tmp:c}
	c:=\lambda_{\min} (-\bfH)>0.
\end{equation}
With this choice of $\bfH$ and
with the length scales $l_s=2$, $l_t=4$,
 the condition in~\eqref{tmp:c} is not satisfied,
 suggesting that the problem is not block-wise log-concave. 
 
\subsubsection{MCMC samplers}
We apply simplified manifold MALA (MMALA) \cite{GC11},
and MMALA-within-Gibbs to draw samples from the posterior distributions
of the $16\times 16$, $32\times 32$ and $64\times 64$ problems. 
We also apply pCN, as an example of a function-space MCMC scheme,
to illustrate that function-space MCMC is not dimension independent when some of its underlying assumptions are not met.

The MMALA proposal is
\begin{equation*}
	\tilde{\bfx}^k = \bfx^k+\tau\bfM\, \nabla \log p(\bfx^k\vert \bfy)+\sqrt{2\tau}\,\bfM^{1/2} \xi^{k+1},\quad \xi^{k+1}\sim\mathcal{N}(\bf{0},\bf{I}),
\end{equation*}
where the choice $\bfM = \bf{\Lambda} +\bfB^{-1}$ turns MALA into 
simplified manifold MALA.
The matrix ${\bf{\Lambda}}$ is diagonal and the $i$th diagonal element is
$[{\bf{\Lambda}}]_{i,i} = \exp(\mu + [\bfB]_{i,i})$;
see \cite{GC11}.
We implement MMALA-within-Gibbs using blocks of size $q=d\times d$ and
consider $d=8,16,32,64$.
We emphasize that MMALA-within-Gibbs with a single block,
covering the entire domain, is equivalent to the MMALA sampler.
For example, if $L=64$ and $d=64$, the sampler does not use partial updating
and we recover the usual MMALA;
with $L=64$ and $d=16$, we divide the domain into 16 blocks,
each of size $16\times 16$.
The blocks define a neighborhood of components $X_{i,j}$
of size $d\times d$ on the 2D-domain and are ordered left-to-right and top-to-bottom.

All samplers are initialized at the maximum a posteriori point (MAP)
which we find by solving the optimization problem
\begin{equation*}
	\min_\bfx -\log \pi(\bfx\vert \bfy),
\end{equation*}
using a Gauss--Newton method.
We consider various step sizes $\tau$
and, for each one, we run pCN to generate $10^5$ samples
and MMALA or MMALA-within-Gibbs to generate $10^4$ samples.
We then compute the integrated auto correlation time (IACT) 
of each pixel using the techniques described in \cite{Wolff04}.
Note that we use all samples (no burn-in)
to compute the average acceptance ratios and IACT.
We inspected some of the chains and could not identify 
an apparent transient phase, likely because our initialization
point makes the transients negligible.

%

\subsubsection{MCMC results}
Results of a MMALA-within-Gibbs sampler
with $d=8$ and step size $\tau=0.5$ 
are shown in Figure~\ref{fig:LogGaussIllu}.
\begin{figure}[tb]
\begin{center}
\includegraphics[width=0.8\textwidth]{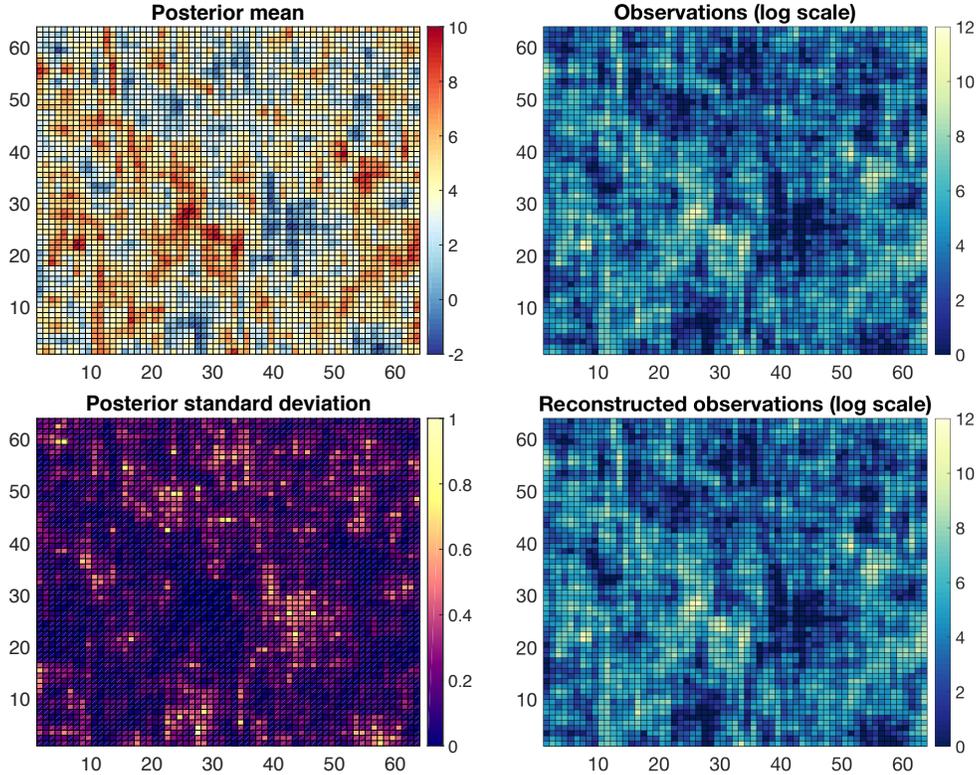}
\caption{Illustration of results obtained by
$10^4$ samples of a MMALA-within-Gibbs sampler
with $d=8$, and step size $\tau=0.5$.
Top row:  posterior mean (left) and observations $Y_{i,j}$ (right).
Bottom row: posterior variance at each grid point (left),
observations corresponding to posterior mean (right).}
\label{fig:LogGaussIllu}
\end{center}
\end{figure}
The panels in the top row show the
posterior mean (average of all MCMC samples)
and the observations $Y_{i,j}$ (on a log-scale).
The panels in the bottom row show the posterior variance at each grid point 
and the observations (on a log-scale) corresponding to the posterior mean.
We note a good agreement between the posterior mean and the 
true field (see Figure~\ref{fig:Truth}),
as well as a good agreement between the observations and the reconstructed obervations.

Our tuning of the step-size is  illustrated in Figure~\ref{fig:LogGaussAccRate},
where the average acceptance ratio of MMALA and MMALA-within-Gibbs
is plotted as a function of the step sizes we tried for the problem with $L=64$.
\begin{figure}[tb]
\begin{center}
\includegraphics[width=0.8\textwidth]{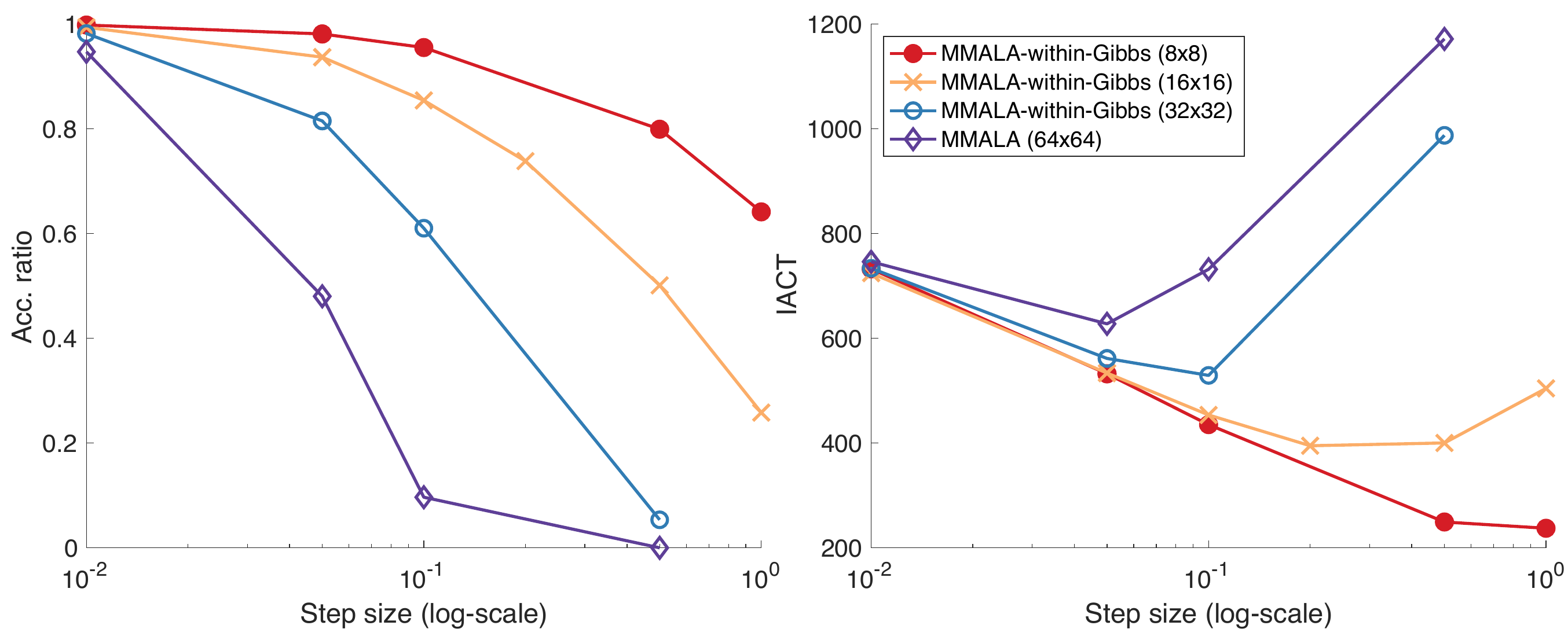}
\caption{Left: average acceptance ratio of MMALA and MMALA-within-Gibbs
as a function of the step size for the $64\times 64$ problem.
Right: IACT of MMALA and MMALA-within-Gibbs
as a function of the step size for the $64\times 64$ problem.}
\label{fig:LogGaussAccRate}
\end{center}
\end{figure}
The results are qualitatively similar for the problems with $L=16$ and $L=32$.
We see that the acceptance ratio decreases with step size,
but for any fixed step size $\tau$, the acceptance ratios
of the within-Gibbs samplers increase when the block sizes are decreased.
The reason is that the partial updating of the MMALA-within-Gibbs 
sampler results in large acceptance ratios for large step sizes
independently of the dimension of the problem.
Figure~\ref{fig:LogGaussAccRate} also shows IACT as a function of the step size.
We note that, for a fixed step size,
IACT increases with the block size 
and that the step size that minimizes IACT 
 increases as the block size decreases. 
 Again, the reason is that the partial updating strategy of MMALA-within-Gibbs
 allows larger steps for smaller blocks, which decreases IACT and 
 accelerates the mixing of the Markov chain.

IACT, averaged over all grid points, 
and the average acceptance probabilities (averages taken over the MCMC moves) 
are listed  in Table~\ref{tab:ResultsCox}.
\begin{table}[h]
\footnotesize
\begin{center}
\begin{tabular}{c ccc cc}
       &   	  & 		   \multicolumn{4}{c}{MMALA-within-Gibbs ($N_e=10^4$)}\\
$L$ & pCN ($N_e=10^5$) & $64\times 64$ blocks & $32\times 32$ blocks & $16\times 16$ blocks& $8\times 8$ blocks\\\hline
16&4626/0.18/0.002&-&-&342/0.95/0.2&204/0.93/0.5\\
32&5363/0.26/0.002&-&437/0.75/0.1&330/0.73/0.2&203/0.78/0.5\\
64&6884/0.21/0.001&627/0.48/0.05&529/0.61/0.1&394/0.75/0.2&249/0.80/0.5
\end{tabular}
\end{center}
\caption{IACT\slash average acceptance probability\slash$\tau$
of pCN, MMALA, and MMALA-within-Gibbs
for three problems with increasing domain size $L$ (and thus increasing dimension).
Note that MMALA-within-Gibbs with block size equal to the domain size corresponds to MMALA.}
\label{tab:ResultsCox}
\end{table}
Here, we list results where we fixed the step size for each considered block size
to make the resulting average acceptance probabilities comparable.
An even better agreement between the acceptance probabilities at each block 
would require a more careful tuning of the step size,
but the step size tuning we carried out is sufficient to make our points
and to illustrate the relevant characteristics of the samplers.

For a fixed block size, MMALA-within-Gibbs yields the same IACT
independently of the domain size (dimension).
For example, with blocks of size $d=16$, IACT of the $4096$-dimensional problem
is similar to the IACT of the $256$- or $1024$-dimensional problems.
Moreover, the step size and corresponding acceptance ratios seem to be 
independent of the overall problem dimension.
The numerical experiments thus corroborate our theoretical results
on dimension independent convergence of MALA-within-Gibbs,
even when the assumption of block-wise log-concavity, required for our proofs, 
is not satisfied with our choices of block size.

Comparing MMALA-within-Gibbs with MMALA and pCN,
we note that the IACT of pCN and of MMALA with $n$-dimensional updates increase with dimension.
In the case of pCN, this behavior is not surprising because the way that dimension 
increases in this sequence of problems breaks some of the assumptions that 
are required for the dimension independence of pCN.
For MMALA, this behavior is also to be expected since MMALA in itself is not dimension independent. 
Using the within-Gibbs framework turns MMALA into a dimension independent sampling algorithm.

The dimension-independent convergence of MMALA-within-Gibbs, however,
does not necessarily imply that MMALA-within-Gibbs is the most efficient
sampler for this problem.
Using the cost-per-effective-sample in Equation~\eqref{eq:CostPerSample},
it is evident that MMALA is more efficient than MMALA-within-Gibbs
(at the block sizes we consider).
The cost estimate~\eqref{eq:CostPerSample}, however, assumes 
that a full likelihood evaluation
is required for each proposed sample of MMALA-within-Gibbs,
which is a conservative estimate.
One can easily envision making use of the problem structure
during likelihood evaluations in each block.
For example, evaluation of the prior term in~\eqref{eq:CoxLogPost}
in each block does not require computing 
the full matrix-vector product $\bfB^{-1/2}(\bfx-\mu\mathbf{1})$.
One can speed up the computations by only updating the relevant 
components that are modified in the current block.
We did not pursue such ideas because this problem is relatively simple
and because our main goal is to demonstrate that MMALA-within-Gibbs
can exhibit dimension independence.

\subsection{Numerical illustration 2: inverse problems with an elliptic PDE}
\label{sec:Subsurface}
We consider the PDE
\begin{equation*}
-\nabla\cdot(\kappa\, \nabla u)=g,
\end{equation*}
on a square domain $(s,t)\in [0,1]^2$
with Dirichlet boundary conditions;
here $u$ represents a pressure field and $g$ is a given source term,
which consists of four delta functions (sources) at four locations in the domain. Details on the boundary conditions and source term are given in  \cite{MTWC15}.
The quantity $\kappa > 0$ represents the permeability of the medium; we use a log-normal prior for the permeability to enforce the non-negativity constraint.
Thus, $K=\log \kappa$ is a Gaussian random field.
We set its mean to be zero and employ the covariance kernel
\begin{equation*}
k(s_1, t_1; s_2, t_2) = \exp\left(
-\frac{(s_1-s_2)^2}{2l_s^2}
-\frac{(t_1-t_2)^2}{2l_t^2}
\right)
\end{equation*}
where $(s_1,t_1)$ and $(s_2,t_2)$ are two points in the square domain and
$l_s$ and $l_t$ are correlation length scales.
Our goal is to estimate the permeability given 128 noisy observations of the pressure $u$
in the center of the domain.
This problem setup is also described in \cite{MTWC15}.
The inverse problems we consider here differ from those in \cite{MTWC15}
only in the correlation lengths of the prior,
which do not affect the numerics of the PDE solve, 
the gradient computations, or the observation and forcing network.
We thus refer to \cite{MTWC15} for the details of the 
numerical solution of the PDE,
and in particular to Figure~2 of \cite{MTWC15} 
for descriptions of the locations of the forcing terms.

\subsubsection{Discretization and problem setups}
For computations, we discretize the PDE using a standard finite element method
with a uniform grid of $16\times 16$ points (see \cite{MTWC15} for details of the discretization).
The discretization leads to the algebraic equation
\begin{equation}
\label{eqn:PDEsubDisc}
\bfA(\hat{\boldsymbol{\kappa}}) \hat{\bfu}=\hat{\bfg},
\end{equation}
where the hat over variables denotes discretized quantities,
i.e., $\hat{\boldsymbol{\kappa}}$, $\hat{\bfu}$, and $\hat{\bfg}$
are vectors of size $N_u=256$ and $\bfA$ is a $256\times 256$ matrix
that depends on the permeability $\hat{\boldsymbol{\kappa}}$.
We will be computing with the discretized PDE from now on
and, for that reason, we drop the hats above all variables.
The pressure observations are modeled by the equation
\begin{equation}
\label{eq:LikelihoodSubsurface}
	\bfy = \bfH\bfu+\boldsymbol{\eta},\quad \boldsymbol{\eta}\sim\mathcal{N}({\bf{0}},\bfR),
\end{equation}
where $\bfH$ is a $N_y\times N_u$ matrix that has exactly one $1$ in each row
and picks out every other component of $\bfu$.
The observation noise covariance is set to be $\bfR=0.1^2\,\bfI$.

After discretization, the log-permeability $\bfK$ is finite dimensional
and its prior distribution is the finite dimensional Gaussian $\mathcal{N}(0,\bfB)$. 
Due to the squared exponential covariance model,
$\bfB$ can be well approximated by a low-rank matrix, i.e.,
\begin{equation}
\bfB\approx \bfU_\theta \bfL_\theta\bfU_\theta^T,
\label{eq:truncprior}
\end{equation}
where $\bfL_\theta$ is a $N_\theta\times N_\theta$
diagonal matrix whose diagonal elements are the $N_\theta < N_u$ largest eigenvalues of $\bfB$
(see \cite{MTWC15} for details).

The Gaussian prior for the log-permeability and the likelihood in~\eqref{eq:LikelihoodSubsurface}
define the posterior distribution $\pi(\bfK\vert \bfy)\propto \pi_0(\bfK)\pi_l(\bfy\vert \bfK)$ for the log-permeability:
\begin{equation*}
	\pi(\bfK\vert \bfy)\propto \exp\left(
	-\frac{1}{2} \|\bfB^{-1/2}\bfK\|^2
	-\frac{1}{2} \|\bfR^{-1/2}(\mathcal{M}(\bfK)-\bfy)\|^2
	\right);
\end{equation*}
here $\mathcal{M}$ maps the log permeability to the pressure at observation locations,
i.e., $\mathcal{M}(\bfK)=\bfH \bfu(\exp(\bfK))$, 
with the $\bfu$ being the solution to the discretized PDE \eqref{eqn:PDEsubDisc}.

Since symmetric positive semi-definite matrices
can always be diagonalized by a coordinate transformation, we consider the change of variables 
\begin{equation}
\label{eq:KToTheta}
	{\boldsymbol{\theta}} = \bfL_{\boldsymbol{\theta}}^{-1/2}\bfU_{\boldsymbol{\theta}}^T\bfK\approx\bfB^{-1/2}\bfK,
\end{equation}
which leads to the posterior distribution
\begin{equation}
\label{eq:ThetaPost}
\pi({\boldsymbol{\theta}}\vert \bfy)\propto\exp\left(
-\frac{1}{2} \|  {\boldsymbol{\theta}}\|^2
-\frac{1}{2} \| \bfR^{-1/2}(\mathcal{M}(\bfK({\boldsymbol{\theta}}))-\bfy)\|^2
\right).
\end{equation}
Below, we use MCMC samplers to draw samples
from the posterior distribution of ${\boldsymbol{\theta}}$.
The corresponding (log-)permeabilities are computed from posterior samples of ${\boldsymbol{\theta}}$
via the inverse of the transformation~\eqref{eq:KToTheta}.

We consider two problem setups, which differ in the correlation lengths of the log-normal prior.
The correlation lengths define the dimension of $\boldsymbol{\theta}$ in that the latter is chosen to retain 95\% of the integrated prior variance.
Specifically, if the correlation lengths are short compared to the $[0,1]\times [0,1]$ domain,
then the dimension is of ${\boldsymbol{\theta}}$ is large;
if the correlation lengths are large, 
the dimension of ${\boldsymbol{\theta}}$ is small.
The correlation length scales
and implied dimensions of ${\boldsymbol{\theta}}$ 
of Setups~1 and~2 are summarized in Table~\ref{tab:ProblemSummary}.
\begin{table}[h!]
\begin{center}
\begin{tabular}{rccc}
&$l_s$ & $l_t$  &$N_{\boldsymbol{\theta}}$  \\\hline
Setup 1& 0.4 & 0.8 & 30\\
Setup 2& 0.2 & 0.1 & 136\\
\end{tabular}
\end{center}
\caption{Correlation lengths and reduced dimensions for Setups~1 and~2.}
\label{tab:ProblemSummary}
\end{table}%

We illustrate the decay of the prior covariance eigenvalues and the true log-permeabilities of Setups~1 and~2
in Figure~\ref{fig:priors}.
\begin{figure}[tb]
\begin{center}
\includegraphics[width=1\textwidth]{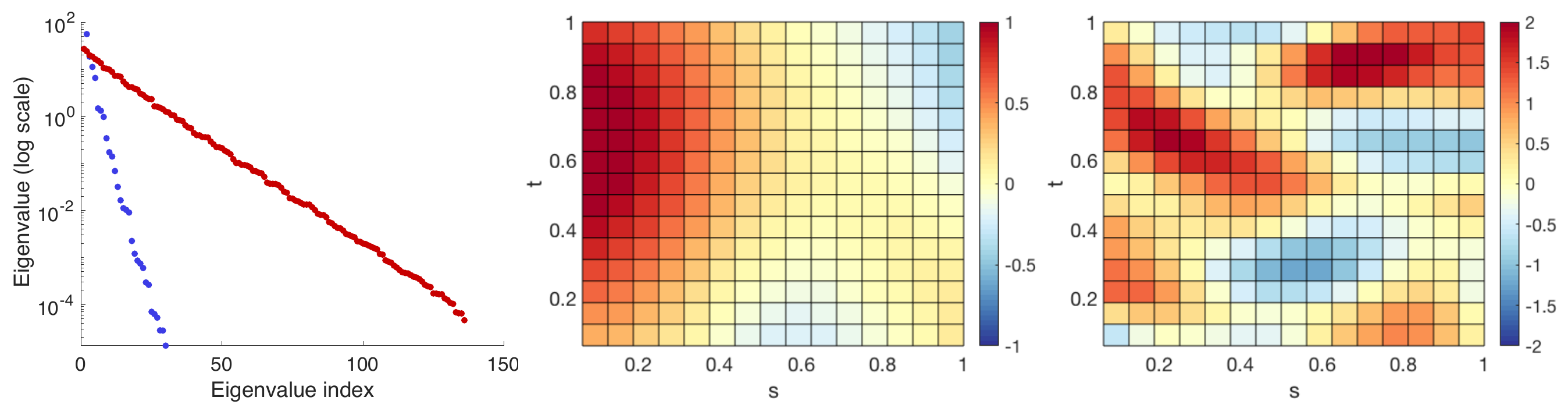}
\caption{
Left: eigenvalues of the prior covariance matrix 
for Setup~1 (blue) and Setup~2 (red).
Center: true log-permeability of Setup~1.
Right: true log-permeability of Setup~2.
}
\label{fig:priors}
\end{center}
\end{figure}
Specifically, we note that the eigenvalues decay more quickly for Setup~1 than for Setup~2 
because Setup~1 is characterized by larger correlation length scales than Setup~2.
The true log-permeabilities of Setups~1 and~2
are random draws from the prior and are shown in the right panels of Figure~\ref{fig:priors}.
There is more small-scale structure in the log-permeability of Setup~2
than in Setup~1, again due to the shorter prior correlation length scales.

Setup~2 is intended to have a higher dimension than Setup~1, not just in the apparent dimension of $\boldsymbol{\theta}$ but also in the sense of the prior-to-posterior update and hence the influence of the data (i.e., the effective dimension, as defined in \cite{Agapiou16}).
We achieve this by keeping the domain size fixed while decreasing the correlation lengths, 
which effectively increases the number of degrees of freedom in the unknown. 
In the previous log-Cox example, we imposed a similar growth by keeping the correlation lengths fixed but increasing the domain size.
Note that, however, in the previous example we also increased the number of observations with the dimension (size of the domain), while in this example, we keep the number of  observations fixed. 
This is a minor issue because in Setup 1, due to the large prior correlation lengths, many of the observations are strongly dependent (i.e., in the prior predictive $\pi(\bfy)$). When the correlation lengths decrease in Setup~2, the number of effectively independent observations increases and thus the relative influence of the likelihood,
and hence the effective dimension, increase as well. 

We emphasize that the way dimension increases in these two elliptic PDE inverse problem setups is 
different from what is usually considered in function-space MCMC.
Dimension independence of function-space MCMC 
requires that the dimension increase due to refinement of a discretization,
while keeping all other aspects of the problem setup (e.g., the prior and forward operator, assuming a consistent discretization scheme) fixed.
As in the previous example, we in fact do the opposite and keep the discretization fixed,
but decrease the prior correlation length scales.
For this reason, one cannot expect that pCN, or other function-space MCMC techniques,
will exhibit dimension independence in the problem setups we consider.

\subsubsection{Sparse conditional structure and block-wise log-concavity}
The theory we created for the dimension independent convergence of the
MALA-within-Gibbs sampler relies on assumptions of sparse conditional
structure and block-wise log-concavity.  With our choice of prior, the
problem does not have sparse conditional structure in 
$(s,t)$-coordinates. Yet the coordinate transformation
\eqref{eq:KToTheta} produces a sparse conditional structure---indeed
complete independence---in the prior for the
$\boldsymbol{\theta}$-coordinates, which correspond to discretized
Karhunen-Lo\`eve (KL) modes.  Conditioning on the observations,
however, can introduce dependence among the smoother KL modes because
an observation at a given $(s,t)$-location is in principle influenced
by all of the modes---due to the nature of the elliptic operator and
the KL modes' global support. Conversely: changes in one KL mode can
affect the solution everywhere in $(s,t)$-coordinates.
Nonetheless, our experiments, along with various other experiments with this problem found in the literature, suggest that this dependence is weak and that the problem thus has an approximate sparse conditional structure in the KL modes.
The assumption of block-wise log-concavity is difficult to verify in this example, in either the $\boldsymbol{\theta}$ or $(s,t)$-coordinates. The reason is that the discretization of the PDE, e.g., in \eqref{eq:ThetaPost},
makes computations difficult because we do not have second-order adjoints to compute the required Hessian.

\subsubsection{MCMC samplers}
We use pCN, MALA, and MALA-within-Gibbs 
to draw samples from the posterior distribution~\eqref{eq:ThetaPost}.
Again, we emphasize that MALA-within-Gibbs with a sufficiently large block size
is the same as the usual MALA without partial updating.
The pCN proposal is
\begin{equation*}
	\tilde{{\boldsymbol{\theta}}}^{k+1} = \sqrt{1-\beta^2}{\boldsymbol{\theta}}^k+\beta \boldsymbol{\xi}^{k+1},
\end{equation*}
where ${\boldsymbol{\theta}}^k$ is the current state of the MCMC
and where $\boldsymbol{\xi} \sim\mathcal{N}({\bf{0}},\bfI_{N_{\boldsymbol{\theta}}})$,
$\bfI_{N_{\boldsymbol{\theta}}}$ being the identity matrix of order $N_{\boldsymbol{\theta}}$.
The proposed $\tilde{{\boldsymbol{\theta}}}^{k+1}$ is accepted with probability 
\begin{equation*}
	\alpha_\text{pCN} = 1\wedge \exp\left(
	 \frac{1}{2}\|\bfR^{-1/2}(\mathcal{M}(\bfK({\boldsymbol{\theta}}^k))-\bfy)\|^2
	- \frac{1}{2}\|\bfR^{-1/2}(\mathcal{M}(\bfK(\tilde{{\boldsymbol{\theta}}}^{k+1}))-\bfy)\|^2
	\right),
\end{equation*}
where $1\wedge x$ denotes $\min\{1,x\}$.
We initialize the pCN chain at the MAP, which we find by
quasi-Newton optimization (Matlab's fminunc) of the cost function
\begin{equation}
\label{eq:F}
F({\boldsymbol{\theta}}) = \log (\pi({\boldsymbol{\theta}}\vert \bfy))=
-\frac{1}{2} \|  {\boldsymbol{\theta}}\|^2
-\frac{1}{2} \| \bfR^{-1/2}(\mathcal{M}(\bfK({\boldsymbol{\theta}}))-\bfy)\|^2
+C,
\end{equation}
where $C$ is a constant that is irrelevant.
We tune the parameter $\beta$
to obtain minimal IACT.
As above, IACT is computed using the techniques and definitions of \cite{Wolff04}.

The MALA proposal for this problem is
\begin{equation*}
	\tilde{{\boldsymbol{\theta}}}^{k+1}= {\boldsymbol{\theta}}^k-\tau \bfJ^{-1} \nabla_{\boldsymbol{\theta}} F({\boldsymbol{\theta}}^k)+\sqrt{2\tau}\bfJ^{-1/2}\boldsymbol{\xi}^{k+1},
\end{equation*}
where $F({\boldsymbol{\theta}})$ is as in \eqref{eq:F}
and where $\bfJ$ is the Hessian of $F$ at the MAP.
As with pCN, we initialize MALA at the MAP and tune the step size $\tau$ of MALA to find a minimal IACT.
As in the previous example, we use all samples for our computations (no burn-in).

MALA-within-Gibbs requires that we partition ${\boldsymbol{\theta}}$ into blocks.
Above, we argued that this problem has an approximate 
sparse conditional structure in the $\boldsymbol{\theta}$ coordinates.
For this reason, we use partitions of ${\boldsymbol{\theta}}$ that 
group consecutive elements of ${\boldsymbol{\theta}}$ together. 
Below, we consider several block sizes
and for each one, we initialize MALA-within-Gibbs at the MAP
(as before) and tune the step size to achieve a minimal IACT.

\subsubsection{MCMC results}
\begin{table}[tb]
\footnotesize
\begin{center}
\begin{tabular}{rc ccc c}
&Method & Length of chain & IACT & Acc. ratio& Step  \\\hline
\parbox[t]{2mm}{\multirow{4}{*}{\rotatebox[origin=c]{90}{Setup~1}}}&MALA-within-Gibbs, $q=1$ & $10^4$ & 25 & 0.43 & 0.5 \\
&MALA-within-Gibbs, $q=15$ & $10^4$ & 141 & 0.22 & 0.05\\
&MALA\slash MALA-within-Gibbs, $q=30$ & $10^5$ & 246 & 0.44 & 0.01 \\
&pCN                                      & $10^6$ & 6,102 & 0.24 & 0.01 \\
\hline\hline
\parbox[t]{2mm}{\multirow{4}{*}{\rotatebox[origin=c]{90}{Setup~2}}}&MALA-within-Gibbs, $q=1$ & $10^3$ & 20 & 0.38 & 0.500 \\
&MALA-within-Gibbs, $q=68$ & $10^4$ & 367 & 0.23 & 0.010\\
&MALA\slash MALA-within-Gibbs, $q=136$ &  $10^5$ & 923 & 0.23 & 0.005\\
&pCN & $10^6$ & 28,015 & 0.45 & 0.010 
\end{tabular}
\end{center}
\caption{Summary of simulation results of Setups~1 and~2.}
\label{tab:ResultsSubsurface}
\end{table}

Typical results one can obtain via MCMC are shown in Figure~\ref{fig:MeanAndStandDiv},
\begin{figure}[tb]
\begin{center}
\includegraphics[width=.7\textwidth]{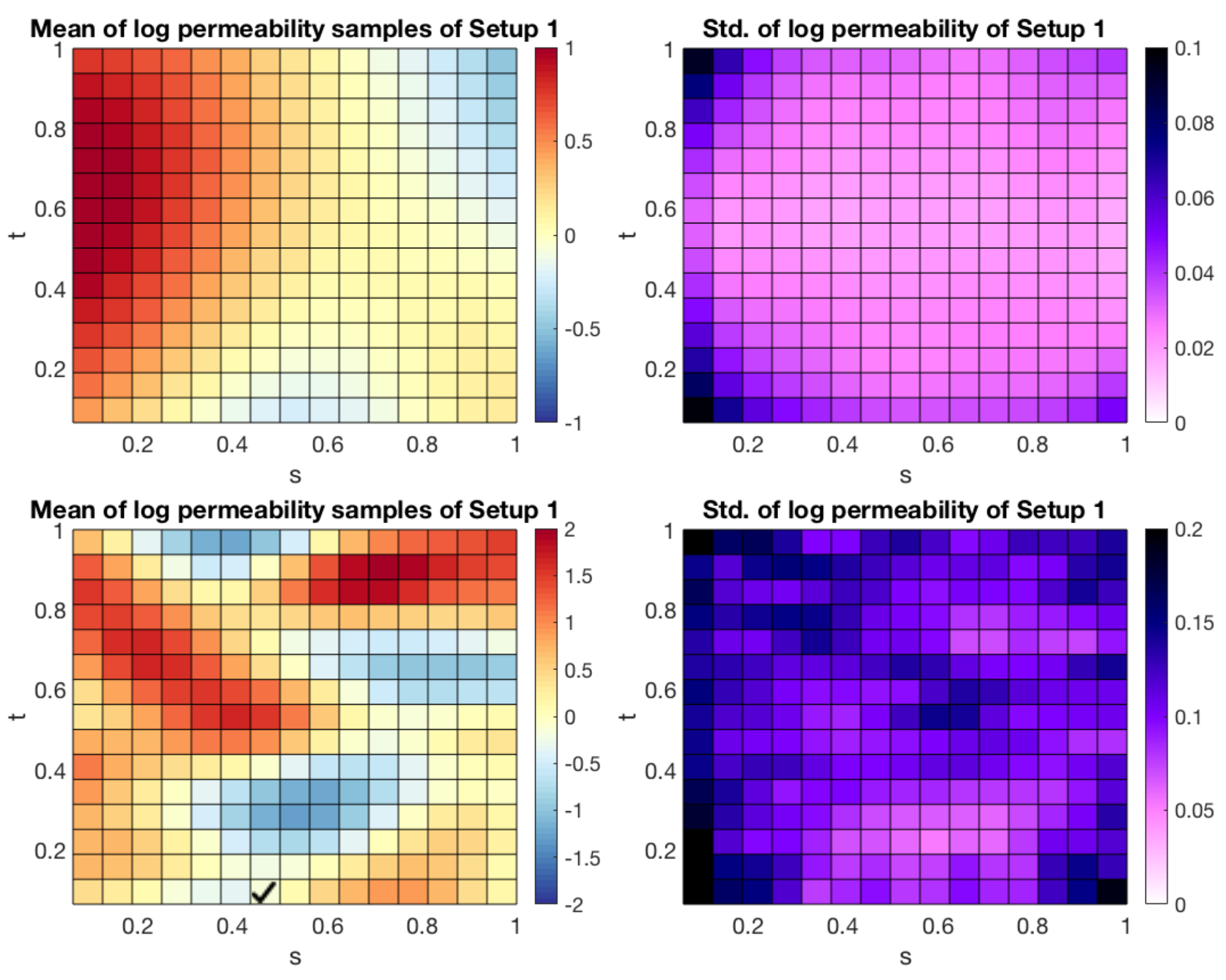}
\caption{Top row: 
approximate posterior mean (left)
and approximate standard deviation (right) 
computed from $N_e=10^4$ MALA-within-Gibbs 
samples with block size $q=1$ for Setup~1.
Bottom row:
approximate posterior mean (left)
and approximate standard deviation (right) 
computed from $N_e=10^3$ MALA-within-Gibbs 
samples with block size $q=1$ for Setup~2.
}
\label{fig:MeanAndStandDiv}
\end{center}
\end{figure}
where we plot an approximation of the posterior mean of the log-permeability
and the approximate posterior standard deviations (on the grid)
computed via MALA-within-Gibbs.
We obtain an approximate posterior mean of the log-permeability, $\bfK$,
from the posterior mean of ${\boldsymbol{\theta}}$,
by mapping ${\boldsymbol{\theta}}$ to $\bfK$ via the inverse of~\eqref{eq:KToTheta}.
The approximate posterior mean of $\bfK$ should be compared to the true
log-permeability in Figure~\ref{fig:priors}.

A summary of the numerical experiments we performed
is provided in Table~\ref{tab:ResultsSubsurface}.
The table lists IACT,
step sizes, and average acceptance ratios for the various MCMC samplers.
The numbers shown are tuned,
in the sense that we only show results for the step size
that leads to minimal IACT (over all step sizes we tried).

We note that the IACT of MALA-within-Gibbs with block size one
is nearly identical for the two problem setups,
indicating that the dimension independence results we obtained under 
more restrictive assumptions may indeed hold in practice.
As in the previous example, we also note that the step size $\tau$
that leads to minimal IACT
decreases as we increase the size of the blocks of MALA-within-Gibbs.
This is further illustrated in Figure~\ref{fig:Subsurface_AccRatio},
\begin{figure}[tb]
\begin{center}
\includegraphics[width=0.8\textwidth]{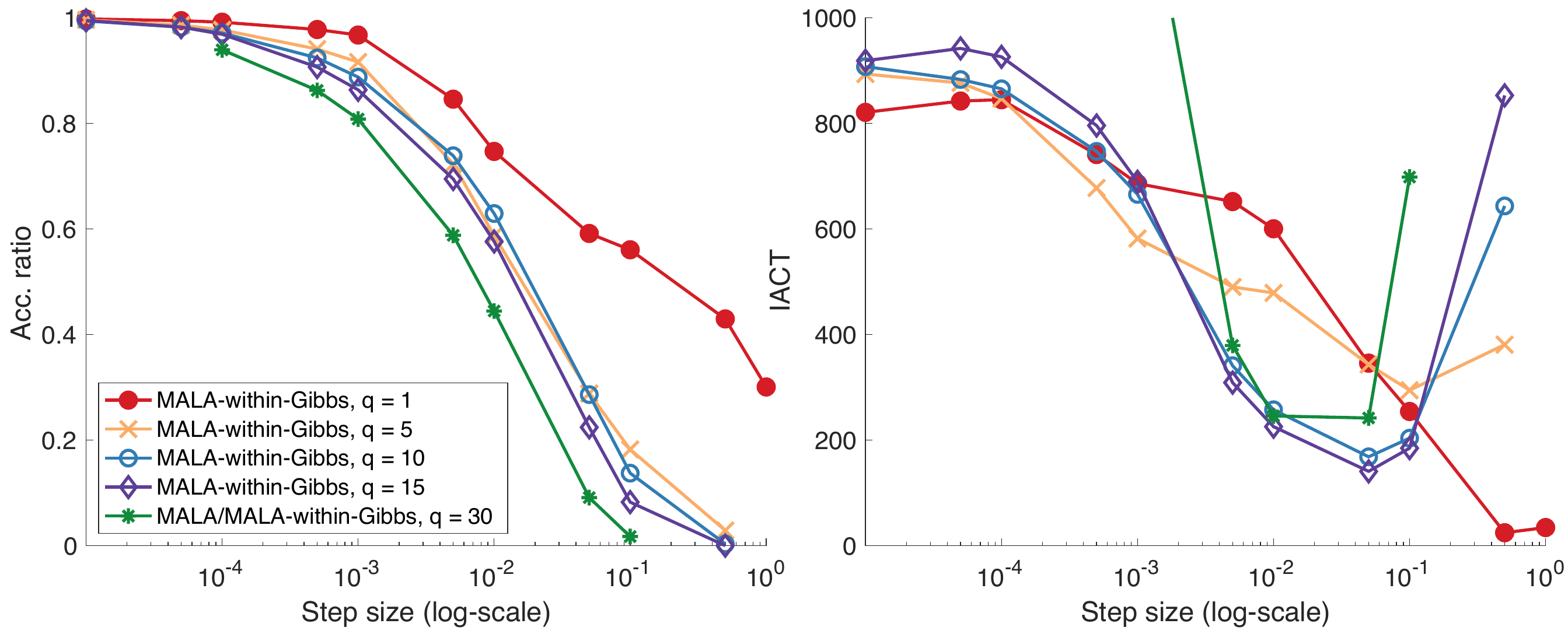}
\caption{Left: average acceptance ratio of MALA-within-Gibbs and MALA,
as a function of step size, for Setup~1.
Right: average IACT of MALA-within-Gibbs and MALA,
as a function of step size, for Setup~1.}
\label{fig:Subsurface_AccRatio}
\end{center}
\end{figure}
where we plot the average acceptance ratio as a function of the step size
for MALA-within-Gibbs (several block sizes) and MALA.
As in the previous example, we note that for a given fixed step size,
the average acceptance ratio increases as we decrease the block size.
The figure also shows IACT as a function of the step size for the various samplers, with optimal step sizes clearly visible.
We note, as before, that the partial updating of MALA-within-Gibbs pushes
the step size that minimizes IACT towards larger values.

We also note that the IACT of pCN is larger than that of MALA,
and that the IACT of pCN and MALA increase with the dimension of ${\boldsymbol{\theta}}$.
The reason that the IACT of pCN is larger than that of MALA may be that 
MALA makes use of gradients of the target, while the pCN proposal does not directly exploit gradient or likelihood information. Increase of the IACT of pCN with dimension is due to the way in which dimension increases
when going from setup~1 to setup~2.
As in the previous example (and as explained above), 
some of the fundamental assumptions that are needed for dimension independence of pCN
are not satisfied when transitioning from setup~1 to setup~2.
For that reason, one cannot expect that pCN be dimension independent in the scenario we consider here.

Finally, we note that the acceptance rate of pCN that minimizes IACT (over the step sizes we tried) in Setup~2
is substantially larger than optimal acceptance rates of RWM (45\% rather than about 20\%).
This is in line with recent numerical experiments and analyses which suggest 
that for Gaussian targets, a good acceptance rate may be near 50\%.
We make no claim, however, that our tuning of pCN is perfect.
We considered a wide range of step sizes and ran pCN chains of length $10^6$ for each choice. One could possibly achieve a slightly better IACT with further tuning, but nonetheless the IACT of pCN can be expected to be significantly larger than that of MALA  or MALA-within-Gibbs.
Moreover, given the overall chain length of only $10^6$, the estimated IACT of $28015$ for pCN may not be entirely precise, but all of our numerical experiments indicate that it is in any case very large.

Recall that a small and dimension-independent IACT of MALA-within-Gibbs does not 
necessarily imply that the algorithm is a computationally efficient sampler;
generating one sample requires several likelihood evaluations
due to the partial updating strategy, as described above.
Estimating the cost per effective sample by~\eqref{eq:CostPerSample},
we see that MALA-within-Gibbs is not an efficient sampler for Setup~1,
but MALA-within-Gibbs with $q=68\times 68$ 
(leading to two blocks and two likelihood evaluations per sample)
is indeed the most effective sampler for Setup~2.
This further illustrates that there is a trade-off between the need to reduce IACT
by using partial updating and the need to keep the cost-per-sample,
which we take to be proportional to the number of blocks, reasonable.
In the future, such issues may be addressed by incorporating the 
partial updating into local likelihood evaluations which
do not require solving the full PDE, similar to the localization of forward dynamics discussed in \cite{LT19}, but such issues are beyond the scope of this paper.

We also note that our numerical experiments are limited
in the sense that we only considered pCN, MALA, and MALA-within-Gibbs.
Other samplers may turn out to be more practical than MALA-within-Gibbs.
Specifically, note that the pressure field is relatively well observed, 
which implies that the posterior differs strongly from the prior in many directions
(high effective dimension). 
This explains, at least in part, why we observe such large IACT for pCN.
Other anisotropic samplers that are modifications of pCN, 
e.g., DILI \cite{CuiEtAl16b}, pCNL \cite{Cotter13}, or generalized pCN \cite{rudolf2018generalization},
might be effective in this problem.
Our goal, however, is not to find the most appropriate sampler for this 
Bayesian inverse problem, but rather to use this example
to demonstrate some of the practical and theoretical aspects of the MALA-within-Gibbs sampler.

\section{Conclusion}
Markov chain Monte Carlo (MCMC) samplers are used to draw samples from a given target probability distribution in a wide array of applications. We have discussed the numerical efficiency of a particular sampler, the MALA-within-Gibbs sampler, when the target distribution exhibits a particular sparse conditional structure. In simple terms, the latter is just a structured conditional independence relationship, or block conditional independence relationship, among the variables of interest.
For Gaussians, sparse conditional structure is equivalent to a (block-)sparse precision matrix.
MALA-within-Gibbs samplers are natural tools to make effective use of
sparse conditional structure for numerical efficiency via a suitable partial updating.
We have shown that the acceptance ratio and step size of MALA-within-Gibbs are independent of the overall
dimension of the problem if the partial updating is chosen to be in line
with the sparse conditional structure of the target distribution. 
Under additional assumptions of block-wise log-concavity, 
we could prove that the convergence rate of MALA-within-Gibbs is independent of dimension.
This suggests that MALA-within-Gibbs can be an effective sampler for high-dimensional problems.

We have investigated the applicability of MALA-within-Gibbs in the context of Bayesian inverse problems,
where we expect to encounter sparse conditional structure.
In many Bayesian inverse problems, we expect that sparse conditional structure can be
anticipated based on the prior distribution and the locality of the likelihood, in appropriate coordinates.
Numerical experiments on two well-known test problems suggest that
our theoretical results are indeed indicative of what to expect in practice,
where the required assumptions may only hold approximately.
For example, in both numerical examples, we could show that measures of performance of the 
MALA-within-Gibbs sampler, e.g., integrated autocorrelation time (IACT), step size, and acceptance ratio,
are indeed independent of the overall dimension of the problem.
Nonetheless, the actual computational cost of MALA-within-Gibbs is
dependent on the problem dimension because the partial updating requires
repeated likelihood evaluations (which are costly) per sample.
Our numerical experiments suggest that there is a trade-off between
additional computational costs due to the partial updating and 
the increase in computational cost due to larger IACT or decreasing step size,
without partial updating.
To keep, for example, IACT small, a large number of partial updates should be used,
but this in turn requires several likelihood evaluations per sample.
This trade-off may not always be easy to resolve in practice.
We have provided examples in which MALA-within-Gibbs leads to significant gains
in the computational cost per (effective) sample,
but we have also encountered examples in which a global update is,
ultimately, the right choice.

\bibliographystyle{plain}
\bibliography{References}

\appendix

\section{Proofs}
\label{sec:Proofs}
In this appendix, we provide the proofs of
Proposition \ref{prop:BlockAcceptance} and Theorem \ref{thm:unbiased}.
The proof strategy relies on a maximal coupling of a pair of MALA-within-Gibbs iterations, say $\bfx^k$ and $\bfz^k$.
This convergence is independent of the initial distributions of $\bfx^k$ and $\bfz^k$, 
and we pick a $\bfz^0$ as a sample from $\pi(\cdot)$.
Since $\pi(\cdot)$ is the stationary distribution of MALA-within-Gibbs,
the distribution of $\bfz^k$ remains $\pi(\cdot)$ for all $k$, 
while $\bfx^k$ converges to it. 
Briefly, our proofs consist of four steps.
\begin{enumerate}[i)]
\item
In Section~\ref{sec:Coupling},
we discuss coupling of a pair of MALA-within-Gibbs iterations.
It will become important to consider how the coupled MALA-within-Gibbs iterations, 
$\bfx^k$ and $\bfz^k$, are accepted or rejected and we distinguish the cases
(\emph{i}) accept $\bfx^k$ and $\bfz^k$;
(\emph{ii}) accept $\bfx^k$, reject $\bfz^k$;
(\emph{iii}) reject $\bfx^k$, accept $\bfz^k$;
and (\emph{iv}) reject $\bfx^k$ and $\bfz^k$.
\item 
In Section \ref{sec:analysis},
we derive order $\tau$ estimates of the accept\slash reject probabilities,  
summarized by Lemma \ref{lem:couplingprob}. 
Proposition \ref{prop:BlockAcceptance} follows as a corollary.
\item 
In Section~\ref{sec:BlockUpdates},
we study the dynamics of
the ``block-update'' distance $\|\Delta^k_j\|=\|\bfx^{k}_j-\bfz^{k}_j\|$.
\item
In Section~\ref{sec:proofoftheorem}, it is shown that 
$\|\Delta^k_j\|$ defines a contraction
under the additional assumption of block-wise log-concavity,
which implies that the  sequence $(\|\Delta^k_1\|,\cdots,\|\Delta^k_m\|)$
converges to zero uniformly.
Altogether, this proofs Theorem~\ref{thm:unbiased}.
\end{enumerate}

We will use symbols such as  $M,M_1,M_2$ to denote constants that are independent of
the number of blocks, $m$, or the overall dimension, $n$. 
The constants, however, may depend on the dimension of a block, $q$, 
the sparsity parameter $S$ and other parameters defined in associated assumptions. 
The values of the constants $M,M_1,M_2$ may be different in different places.
We re-use $M,M_1$, and $M_2$ to avoid introducing many different symbols.

\label{sec:proof}
\subsection{Coupling block movements}
\label{sec:Coupling}
Recall that $\bfx^k$ denotes iterates of the MALA-within-Gibbs algorithm, with $\bfx^0$ sampled from a certain initial distribution. 
Now we consider another sequence of iterations of MALA-within-Gibbs, denoted by $\bfz^k$. The initial distribution of  $\bfz^0$ is set to be the target distribution $\pi(\cdot)$.
Since $\pi(\cdot)$ is the stationary distribution of MALA-within-Gibbs,
the distribution of $\bfz^k$ is $\pi$ for all $k$.  

To discuss the block updates within each Gibbs iteration, we use 
\[
\bfx^{k,j}=[\bfx^{k+1}_1,\cdots,\bfx^{k+1}_{j-1}, \bfx^{k}_{j}, \cdots, \bfx^{k}_{m}], \quad
\bfz^{k,j}=[\bfz^{k+1}_1,\cdots,\bfz^{k+1}_{j-1}, \bfz^{k}_{j}, \cdots, \bfz^{k}_{m}], 
\]
to denote the state of the $k$-th Gibbs cycle before the MALA update of the $j$-th block. With this notation,  the $i$-th block of $\bfx^{k,j}$, denoted by $\bfx^{k,j}_i$, is $\bfx^{k+1}_i$ if $i<j$ and is $\bfx^{k}_i$ if $i\geq j$. 

The $j$-th block proposal made to $\bfx^{k,j}$ is given by \eqref{eqn:MALApropose}, and likewise for $\bfz^{k,j}$. We consider coupling the random noises in the two proposals, so they share the same $\xi^k_j$. In other words, the proposals are 
\[
\bfxtilde^k_{j}=\bfx^k_{j}+\tau  \bfv_j (\bfx^{k,j})+\sqrt{2\tau } \xi_{j}^k,\quad \bfztilde^k_{j}=\bfz^k_{j}+\tau  \bfv_j (\bfz^{k,j})+\sqrt{2\tau } \xi_{j}^k. 
\]
We combine them with other blocks from $\bfx^{k,j}$ and $\bfz^{k,j}$, and define
\[
\bfxtilde^{k,j}:=[\bfx^{k+1}_1,\cdots,\bfx^{k+1}_{j-1}, \bfxtilde^{k}_{j}, \bfx^k_{j+1},\cdots, \bfx^{k}_{m}],\quad
\bfztilde^{k,j}:=[\bfz^{k+1}_1,\cdots,\bfz^{k+1}_{j-1}, \bfztilde^{k}_{j}, \bfz^k_{j+1},\cdots, \bfz^{k}_{m}].
\]

The probabilities with which the proposals
$\bfxtilde^{k,j}$ and $\bfztilde^{k,j}$ are accepted are $\alpha_j(\bfx^{k,j},\bfxtilde^{k,j})$ and $\alpha_j(\bfz^{k,j},\bfztilde^{k,j})$ respectively. 
The accept\slash reject step is equivalent
to comparing $\alpha_j(\cdot)$ with a random variable,
uniformly distributed over $[0,1]$.
Specifically, let  $U^k_{j,\bfx}$ be a draw from a uniform distirbution.
The proposal $\bfxtilde^{k,j}$ is accepted if $U^k_{j,\bfx}\leq \alpha_j(\bfx^{k,j},\bfxtilde^{k,j})$. Similarly, the proposal $\bfztilde^k_j$ is accepted
if $U^k_{j,\bfz}\leq \alpha_j(\bfz^{k,j},\bfztilde^{k,j})$,
where $U^k_{j,\bfz}$ is a draw from a uniform distribution.
A maximal coupling of the acceptance steps is achieved by setting $U^k_{j,\bfx}=U^k_{j,\bfz}=U^k_{j}$. 
More specifically, there are four scenarios for the acceptance, based on the value of $U_j^k$:
\begin{equation}
\label{eqn:couple}
\begin{cases}
\text{Both accept} \quad &\text{if }U^k_j\leq \alpha_j(\bfx^{k,j},\bfxtilde^{k,j})\wedge\alpha_j(\bfz^{k,j},\bfztilde^{k,j}),\\
\text{Both reject}\quad &\text{if }\alpha_j(\bfx^{k,j},\bfxtilde^{k,j})\vee\alpha_j(\bfz^{k,j},\bfztilde^{k,j})<U^k_j,\\
\text{Accept  $\bfztilde$ reject $\bfxtilde$}\quad &\text{if }\alpha_j(\bfx^{k,j},\bfxtilde^{k,j})\wedge\alpha_j(\bfz^{k,j},\bfztilde^{k,j})<U^k_j\leq \alpha_j(\bfz^{k,j},\bfztilde^{k,j}), \\
\text{Accept  $\bfxtilde$ reject $\bfztilde$}\quad &\text{if }\alpha_j(\bfx^{k,j},\bfxtilde^{k,j})\wedge\alpha_j(\bfz^{k,j},\bfztilde^{k,j})<U^n_j\leq \alpha_j(\bfx^{k,j},\bfxtilde^{k,j}). 
\end{cases}
\end{equation}
Here, ``accept'' means to set $\bfx^{k+1}_j=\bfxtilde^k_j$
and ``reject'' means to set  $\bfx^{k+1}_j=\bfx^k_j$, 
and likewise for $\bfztilde$;
moreover, $a\wedge b:=\min\{a,b\}$, and $a\vee b:=\max\{a,b\}$. 
It is straightforward to verify that marginally $\bfx^{k+1}_j$ and $\bfz^{k+1}_j$ follow the same distribution (as described in the MALA-within-Gibbs algorithm). 

The information before the update of $\bfx^k_j,\bfz^{k}_j$ is given by the filtration
\[
\mathcal{F}_{k,j}=\sigma\{\bfz^0,\bfx^0,\xi^{t-1}_i,U^{t-1}_i,\xi^t_s,U^t_s, t\leq k-1,s\leq j-1,i=1,\ldots,m\}.  
\]
For simplicity we write $\mathcal{F}_{k}:=\mathcal{F}_{k,1}$, which is the information available when the $k$-th Gibbs cycle starts. It is clear that $\bfx^k, \bfz^k\in \mathcal{F}_k$. We denote the conditional expectation (probability) w.r.t. $\mathcal{F}_{k,j}$ and $\mathcal{F}_k$ as $\E_{k,j} (\Prob_{k,j})$ and $\E_{k} (\Prob_{k})$ respectively.

\subsection{Block-acceptance probabilities (proof or Proposition~\ref{prop:BlockAcceptance}) }
\label{sec:analysis}
We first derive order estimates of the accept\slash reject probabilities, with respect to $\tau$,
by calculating the derivatives of the acceptance probability. 
We do so by establishing a Lemma for a function $f$ such that
\[
\alpha_j(\bfx^{k,j},\bfxtilde^{k,j})=f(\bfx^{k,j}, \sqrt{\tau}, \xi^k_j)\wedge 1,
\]
where $\alpha_j(\bfx^{k,j},\bfxtilde^{k,j})$ is the acceptance probability of the
$j$th block in a MALA-within-Gibbs iteration.
The Lemma is the used to prove Proposition~\ref{prop:BlockAcceptance}.
We will   make repeated use of the fact that 
$\frac{\nabla_{\bfx} \pi(\bfx) }{\pi(\bfx)}=\nabla_{\bfx} \log \pi(\bfx)=\bfv (\bfx)$.
We simplify the notation by writing $w=\sqrt{\tau}$.

\begin{lem}
\label{lem:Lipacc}
Suppose Assumptions \ref{aspt:sparse} and  \ref{aspt:regular} hold. Fix a $j\in \{1,\cdots,m\}$
and, for any given $\bfx\in \reals^{n},\xi\in\reals^{q_j},q_j=\text{dim}(\bfx_j),w\in [0,1] $, define
\[
f(\bfx,w,\xi):=\frac{\pi(\bfxtilde) \exp(-\frac1{4w^2}\|\bfx_j-\bfxtilde_j-w^2 \bfv_j(\bfxtilde)\|^2)}
{\pi(\bfx)\exp(-\frac1{4w^2}\|\bfxtilde_j-\bfx_j-w^2 \bfv_j(\bfx)\|^2)}.
\]
Here the blocks of $\bfxtilde$ are given by  $\bfxtilde_j=\bfx_j+w^2 \bfv_j(\bfx)+\sqrt{2}w \xi$
(see Equation~\eqref{eqn:MALApropose}), and $\bfxtilde_i=\bfx_i$ for $i\neq j$. 
Then 
\begin{enumerate}[(1)]
\item If $i\notin \calI_j$,  $\nabla_{\bfx_i}f(\bfx,w,\xi)=\bf{0}.$
\item There is a constant $M$ such that $\|\nabla_{\bfx} f(\bfx,w,\xi)\|\leq w^2M(\|\xi\|^2+1) f(\bfx,w,\xi)$.
\item There is a constant $M$ such that $|\partial_w f(\bfx,w,\xi)|\leq M(\|\xi\|^2+1) f(\bfx,w,\xi). $
\end{enumerate}
\end{lem}

\begin{proof}
Note that $f$ can be rewritten as 
\begin{align*}
f&=\frac{\pi(\bfxtilde)}{\pi(\bfx)}\exp\left(-\frac1{4} \|w\bfv_j(\bfx)+w\bfv_j(\bfxtilde)+\sqrt{2}\xi\|^2+\frac12 \|\xi\|^2\right)\\
&=\frac{\pi(\bfxtilde)}{\pi(\bfx)}\exp\left(-\frac1{4} (w^2\|\bfv_j(\bfx)\|^2+w^2\|\bfv_j(\bfxtilde)\|^2+2w^2 \bfv_j(\bfx)^T\bfv_j(\bfxtilde)+2\sqrt{2}w \bfv_j(\bfx)^T\xi+2\sqrt{2}w \bfv_j(\bfxtilde)^T \xi)\right),
\end{align*}
where $\bfv_j$ is a $q_j$-dimensional vector.
The fact that the dimension of $\bfv_j$ is $q_j$,
rather than $n$, makes derivations cumbersome.
We thus pad $\bfv_j$ with zero blocks
to form an $n$-dimensional vector $\bfvtilde_j$,
where the $j$-th block of $\bfvtilde_j$ is equal to $\bfv_j(\bfx)$,
but all other components are zero.
Similarly, we form an $n$-dimensional $\tilde{\xi}$ from a $q_j$-dimensional $\xi$
so that the $j$th block of $\tilde{\xi}$ is equal to $\xi$,
but all other components are zero.
With this notation,
\[
\bfxtilde=\bfx+w^2 \bfvtilde_j(\bfx)+\sqrt{2} w \tilde{\xi}.
\]
The associated Jacobians are given by 
\[
\nabla_\bfx \bfxtilde= \bfI+ w^2 \nabla_\bfx \bfvtilde_j(\bfx),\quad \nabla_\bfx \bfvtilde_j(\bfxtilde)=\nabla_{\bfxtilde} \bfvtilde_j(\bfxtilde)(\bfI+ w^2 \nabla_\bfx \bfvtilde_j(\bfx))
\]
Since, by construction, $\|\bfv_j\|^2=\|\bfvtilde_j\|^2$, we have 
\begin{align*}
\log f&=\log \pi(\bfxtilde)-\log \pi(\bfx)\\
&\quad-\frac1{4} (w^2\|\bfvtilde_j(\bfx)\|^2+w^2\|\bfvtilde_j(\bfxtilde)\|^2+2w^2 \bfvtilde_j(\bfx)^T\bfvtilde_j(\bfxtilde)+2\sqrt{2}w \bfvtilde_j(\bfx)^T\tilde{\xi}+2\sqrt{2}w \bfvtilde_j(\bfxtilde)^T \tilde{\xi}).
\end{align*}
The chain rule then gives the gradient of $f$:
\begin{align}
\notag
\frac{\nabla_\bfx f}{f}=\nabla_\bfx \log f&= (\bfI+w^2 \nabla_\bfx \bfvtilde_j(\bfx) )^T \bfv(\bfxtilde)-\bfv(\bfx)
-\frac12 w^2   (\nabla_\bfx \bfvtilde_j (\bfx) )^T \bfvtilde_j(\bfx)\\
\notag
&\quad-\frac12 w^2   (\bfI+w^2\nabla_\bfx \bfvtilde_j(\bfx))^T(\nabla_{\bfxtilde} \bfvtilde_j (\bfxtilde))^T\bfvtilde_j(\bfxtilde)\\
\notag
&\quad-\frac12w^2  \nabla_\bfx \bfvtilde_j(\bfx)^T \bfvtilde_j(\bfxtilde)-\frac12w^2    (\bfI+w^2 \nabla_\bfx \bfvtilde_j(\bfx))^T \nabla_{\bfxtilde} \bfvtilde_j(\bfxtilde)^T\bfvtilde_j(\bfx)\\
\label{tmp:fexp}
&\quad-\frac12 \sqrt{2}w (\nabla_\bfx \bfvtilde_j(\bfx)^T+(\bfI+w^2 \nabla_\bfx \bfvtilde_j(\bfx))^T \nabla_{\bfxtilde} \bfvtilde_j(\bfxtilde)^T) \tilde{\xi}.
\end{align}

We will first verify claim (1). 
Let $\bfu$ be a vector that has nonzero components only in the $i$-th block.
Then, claim~(1) is equivalent to  showing $\bfu^T (\nabla_\bfx \log f) =0$. 
Note the only non-zero blocks of $\nabla_\bfx \bfvtilde_j(\bfx)$ are in its $j$-th row,  
so that only the $j$th block of $\bfu^T \nabla_\bfx \bfvtilde_j(\bfx)^T$ is nonzero.
This block can be written as 
$\bfu_i^T\nabla_{\bfx_i} \bfvtilde_j(\bfx)^T$. 
By Assumption \ref{aspt:regular}, and since $i\notin \calI_j$,
\[
\nabla_{\bfx_i}\bfvtilde_j(\bfx)=\nabla_{\bfx_i}\nabla_{\bfx_j}\log \pi=\mathbf{0}\quad\Rightarrow\quad \bfu^T\nabla_\bfx \bfvtilde_j(\bfx)^T=\mathbf{0}^T.
\]
Knowing this can simplify the computation of $\bfu^T(\nabla_\bfx \log f)$, since most terms in \eqref{tmp:fexp} include the Jacobian matrix, and they drop out when multiplying with $\bfu^T$. The gradient of $f$ in~\eqref{tmp:fexp}
now simplifies to $\bfu^T(\nabla_\bfx \log f) =\langle\bfu,\bfv(\bfxtilde)-\bfv(\bfx)\rangle$,
where we use $\langle a, b\rangle=a^T b $ to denote an inner product. 
Note that $\bfxtilde$ differs from $\bfx$ only in its $j$th block.
Writing $\Delta=\bfxtilde-\bfx\in \reals^{qm}$, we obtain
\begin{align*}
\langle\bfu,\bfv(\bfxtilde)-\bfv(\bfx)\rangle&=\left\langle\bfu,\int^1_0\nabla_\bfx \bfv(\bfx+s\Delta) ds\Delta\right\rangle=0,
\end{align*}
because $\bfu$ has nonzero entries only outside the $j$-th block, and $\Delta$ has nonzero entries only in $j$-th block, and by Assumption \ref{aspt:regular}, the $(\bfx_i,\bfx_j)$-th block of $\nabla_\bfx \bfv=\nabla_{\bfx}^2 \log \pi$ is zero.

For claim (2), we collect terms of  the same $w$ order in \eqref{tmp:fexp}:
\[
\frac{\nabla_\bfx f}{f}=(\bfv(\bfxtilde)-\bfv(\bfx))-\frac12\sqrt{2}w (\nabla_\bfx \bfvtilde_j(\bfx)+\nabla_{\bfxtilde}\bfvtilde_j(\bfxtilde))^T \tilde{\xi}+w^2 R (\xi, \bfx,w).
\]
By Assumption \ref{aspt:regular}, the elements of all vectors and matrices
appearing above are bounded, so  the residual term can be bounded by $\|R(\xi,\bfx,w)\|\leq (\|\xi\|+1) M_1$ with a constant $M_1$. 
By adding and then subtracting a term, we have the following bound
\begin{align}
\label{tmp:nabf1}
\left\|\frac{\nabla_\bfx f}{f}\right\|\leq &\|\bfv(\bfxtilde)-\bfv(\bfx)-\nabla_{\bfx} \bfvtilde_j(\bfx)^T(\bfxtilde-\bfx)\|+w^2(\|\xi\|+1) M_1\nonumber \\
&+\left\|\frac12\sqrt{2}w (\nabla_{\bfx} \bfvtilde_j(\bfx)+\nabla_{\bfxtilde}\bfvtilde_j(\bfxtilde))^T \tilde{\xi}-\nabla_{\bfx} \bfvtilde_j(\bfx)^T (\bfxtilde-\bfx)\right\|.
\end{align}
To bound the first term on the right hand side of \eqref{tmp:nabf1}, first note  that because the $j$-th row block of $\nabla_{\bfx} \bfv(\bfx)$ and $\nabla_{\bfx} \bfvtilde_j(\bfx)$ are the same, and $(\bfxtilde-\bfx)$ has nonzero entries only in the $j$-th block,  
\[
\nabla_{\bfx} \bfv(\bfx)(\bfxtilde-\bfx)=\nabla_{\bfx} \bfv(\bfx)^T(\bfxtilde-\bfx)=\nabla_{\bfx} \bfvtilde_j(\bfx)^T(\bfxtilde-\bfx),
\]
where the first identity is due to the fact that the Hessian of $\log \pi$, $\nabla_\bfx \bfv(\bfx)$, is symmetric. 
The first term in \eqref{tmp:nabf1} can therefore be bounded by a Taylor expansion of $\bfv$,  followed by Assumption \ref{aspt:regular} and Cauchy inequality, 
\begin{align}
\notag
\|&\bfv(\bfxtilde)-\bfv(\bfx)-\nabla_{\bfx} \bfvtilde_j(\bfx)^T(\bfxtilde-\bfx)\|=\|\bfv(\bfxtilde)-\bfv(\bfx)-\nabla_{\bfx} \bfv(\bfx) (\bfxtilde-\bfx)\|\\
\label{tmp:nabf1sol}
&\leq H_v \|\bfxtilde-\bfx\|^2=H_v \|\bfxtilde_j-\bfx_j\|^2=H_v \|w^2 \bfv_j(\bfx)+\sqrt{2}w\xi\|^2\leq  (4w^2 \|\xi\|^2+2w^4 M_v^2) H_v. 
\end{align}
To bound the second term on the right hand side of \eqref{tmp:nabf1}, note that Assumption \ref{aspt:regular}, combined with the Taylor expansion and Young's inequality, gives
\[
\|\nabla_\bfx \bfvtilde_j(\bfx)-\nabla_{\bfxtilde}\bfvtilde_j(\bfxtilde)\|\leq H_v \|\bfx-\bfxtilde\|\leq (\sqrt{2}w\|\xi\|+w^2 M_v)H_v.
\]
Recall that $\tilde{\xi}$ is equal to $\xi$, padded with zeros so that
 $\|\tilde{\xi}\|=\|\xi\|$.
 Thus, the second term on the right hand side of \eqref{tmp:nabf1} is bounded by 
\begin{align}
\notag
&\left\|\frac12\sqrt{2}w (\nabla_{\bfx} \bfvtilde_j(\bfx)+\nabla_{\bfxtilde}\bfvtilde_j(\bfxtilde))^T \tilde{\xi}-\nabla_{\bfx} \bfvtilde_j(\bfx)^T (\bfxtilde-\bfx)\right\|\\
\notag
&\leq \left\|\sqrt{2}w \nabla_{\bfx} \bfvtilde_j(\bfx)^T \tilde{\xi}-\nabla_{\bfx} \bfvtilde_j(\bfx)^T (\bfxtilde-\bfx)\right\|+\left\|\frac12\sqrt{2}w (\nabla_\bfx \bfvtilde_j(\bfx)-\nabla_{\bfxtilde}\bfvtilde_j(\bfxtilde))^T \tilde{\xi}\right\|\\
\notag
&\leq \left\|\sqrt{2}w \nabla_{\bfx} \bfvtilde_j(\bfx)^T\tilde{\xi}-\nabla_{\bfx} \bfvtilde_j(\bfx)^T(\bfxtilde-\bfx)\right\|+(w^2\|\xi\|^2+w^3 \|\xi\| M_v)H_v\\
\notag
&\leq  \left\|\sqrt{2}w \nabla_{\bfx} \bfvtilde_j(\bfx)^T\tilde{\xi}-\nabla_{\bfx} \bfvtilde_j(\bfx)^T(\sqrt{2}w\tilde{\xi}+w^2 \bfvtilde_j(\bfx))\right\|+(w^2\|\xi\|^2+w^3 \|\xi\| M_v)H_v\\
\label{tmp:nabf2sol}
&\leq w^2 M_vH_v+w^2 H_v \|\xi\|^2 +w^3 \|\xi\| M_vH_v. 
\end{align}
If we replace the terms in \eqref{tmp:nabf1} with the bounds in \eqref{tmp:nabf1sol} and \eqref{tmp:nabf2sol}, and if we use the fact that $w\leq 1$, 
we find that there exists a constant $M$ such that
\[
\left\|\frac{\nabla_\bfx f}{f}\right\|\leq w^2 M(\|\xi\|^2+1),
\]
which is our claim (2). 

For claim (3), 
consider the derivative of $\bfxtilde$ with respect to $w$ 
\[
\|\partial_w \bfxtilde\|= \|2w\bfvtilde_j(\bfx)+\sqrt 2\tilde{\xi}\|\leq \sqrt{2}\|\xi\|+2wM_v,
\]
using, again, that $w\leq 1$.
Moreover, since $\nabla_{\bfxtilde}\bfvtilde_j$ has nonzero blocks only on the $j$-th row,
we have that
\[
\|\nabla_{\bfxtilde}\bfvtilde_j (\bfxtilde)\|\leq \sum_{i\in \calI_j} \|\nabla_{\bfxtilde_i}\bfvtilde_j (\bfxtilde)\|\leq SH_v.
\]
Therefore,
\[
\|\partial_w \bfvtilde_j(\bfxtilde)\|=\|\nabla_{\bfxtilde}\bfvtilde_j (\bfxtilde)\partial_w \bfxtilde\|\leq SH_v(\sqrt{2}\|\xi\|+2wM_v).
\]
Finally, recall that
\begin{align*}
\log f&=\log \pi(\bfxtilde)-\log \pi(\bfx)-\frac1{4} (w^2\|\bfvtilde_j(\bfxtilde)+\bfvtilde_j(\bfx)\|^2+2\sqrt{2}w\tilde{\xi}^T(\bfvtilde_j(\bfxtilde)+\bfvtilde_j(\bfx)) ),
\end{align*}
so that
\begin{align*}
\frac{\partial_w f(\bfx,w,\xi)}{f(\bfx,w,\xi)}&=\left\langle \bfv(\bfxtilde),  2w\bfvtilde_j(\bfx)+\sqrt 2\tilde{\xi}\right\rangle-
\frac12 w\|\bfvtilde_j(\bfxtilde)+\bfvtilde_j(\bfx)\|^2\\
&-\frac12 w^2 (\bfvtilde_j(\bfxtilde)+\bfvtilde_j(\bfx))^T\partial_w \bfvtilde_j(\bfxtilde)-\frac{\sqrt{2}}{2}\tilde{\xi}^T(\bfvtilde_j(\bfxtilde)+\bfvtilde_j(\bfx)) -\frac{\sqrt{2}}{2}w\tilde{\xi}^T\partial_w \bfvtilde_j(\bfxtilde).
\end{align*}
Further, recall that the padding with zeros does not affect the inner product: 
\[
\left\langle \bfv(\bfxtilde),  2w\bfvtilde_j(\bfx)+\sqrt 2\tilde{\xi}\right\rangle=\left\langle \bfvtilde_j(\bfxtilde),  2w\bfvtilde_j(\bfx)+\sqrt 2\tilde{\xi}\right\rangle.
\]
Since every term in the expression of $\frac{\partial_w f(\bfx,w,\xi)}{f(\bfx,w,\xi)}$  is, by Assumption \ref{aspt:regular},  bounded, we can conclude there exists an $M$ such that 
\[
|\partial_w f(\bfx,w,\xi)|\leq  (M+M\|\xi\|^2)f(\bfx,w,\xi).
\]
This leads us to claim (3). 
\end{proof}

\begin{lem}
\label{lem:couplingprob}
Under Assumptions \ref{aspt:sparse} and \ref{aspt:regular}, there exists a constant $M$, such that,  for any $\bfx^{k,j}$ and $\bfz^{k,j}$, coupled by one step of MALA-within-Gibbs algorithm at block $j$ as in \eqref{eqn:couple}, we have that
\begin{enumerate}[(1)]
\item $\Prob_{k,j} (\text{accept both})\geq 1-M \sqrt{\tau },$
\item $\Prob_{k,j} (\text{accept only one})\leq M\tau \sqrt{\sum_{i\in \calI_j} \|\bfx^{k,j}_i-\bfz^{k,j}_i\|^2}. $
\item $\E_{k,j} \|\xi_j^k\|\unit_{\text{accept only one}}\leq M\tau \sqrt{\sum_{i\in \calI_j} \|\bfx^{k,j}_i-\bfz^{k,j}_i\|^2}.$
\end{enumerate}
Proposition \ref{prop:BlockAcceptance} follows immediately by the Tower property since
\[
\E[ \alpha_j(\bfx^k,\bfxtilde^k)]=\E[ \E_{k,j}\alpha_j(\bfx^k,\bfxtilde^k)]\geq \E [\Prob_{k,j} (\text{accept both})].
\]
\end{lem}
\begin{proof}
Since here we are concerned with updating one block, for simplicity of the  notation, we write $\bfx=\bfx^{k,j}, \bfxtilde=\bfxtilde^{k,j}, \bfz=\bfz^{k,j}, \bfztilde=\bfztilde^{k,j}$. Let ``\text{reject }$\bfxtilde$" denote the event that the proposal of $\bfxtilde$ is rejected. We will show below that $\Prob_{k,j} (\text{reject }\bfxtilde)\leq \frac12M \sqrt{\tau }$ for a certain $M$. Thus, 
\[
\Prob_{k,j} (\text{accept both})\geq 1-\Prob_{k,j} (\text{reject }\bfxtilde)-\Prob_{k,j} (\text{reject }\bfztilde)\geq 1-M\sqrt{\tau }. 
\]
Let $f(\bfx,w,\xi)$ be as defined  by  Lemma \ref{lem:Lipacc}. 
Note that 
\[
\Prob_{k,j} (\text{reject }\bfxtilde)=\E_{k,j}  (1-f(\bfx,\sqrt{\tau },\xi^k_j)\wedge 1). 
\]
Next we bound $1-f(\bfx,\sqrt{\tau },\xi^k_j)\wedge 1$. Note that $f(\bfx,0,\xi^k_j)=1$.
For each fixed $\bfx,\xi^k_j$, if $f(\bfx,y,\xi^k_j)\leq 1$ for all $y\in (0,\sqrt{\tau }]$, let $w_1=w_2=\sqrt{\tau}$; otherwise, let
\[
w_1=\inf\{y\in [0,\sqrt{\tau }]: f(\bfx,y,\xi^k_j)> 1\},\quad w_2=\sup\{y\in[w_1,\sqrt{\tau} ]: f(\bfx,y,\xi^k_j)\geq 1\}. 
\]
One can check that the following holds with either $f(\bfx,\sqrt{\tau },\xi^k_j)<1$ or $f(\bfx,\sqrt{\tau },\xi^k_j)\geq 1$
\begin{equation}
\label{tmp:1wedge}
1\wedge f(\bfx,\sqrt{\tau },\xi^k_j)-f(\bfx,0,\xi^k_j)=\int^{\sqrt{\tau }}_{w_2}  \partial_y f(\bfx,y,\xi^k_j) dy+\int^{w_1}_0 \partial_y f(\bfx,y,\xi^k_j)dy.
\end{equation}
Note also that by the definition of $w_1$ and $w_2$, 
\[
f(\bfx,y,\xi^k_j)\leq 1,\quad \forall y\in [0,w_1]\cup[w_2,\sqrt{\tau}].
\] 
Thus, by Lemma \ref{lem:Lipacc} (claim~(3)), there is a constant $M_1$ so that:
\[
|\partial_y f(\bfx,y,\xi^k_j)|\leq (M_1+M_1\|\xi^k_j\|^2)|f(\bfx,y,\xi^k_j)|\leq M_1+M_1\|\xi^k_j\|^2,\quad y\in [0,w_1]\cup[w_2,\sqrt{\tau }].
\]
Applying this upper bound to the integrant in \eqref{tmp:1wedge}, 
and using the fact that $f(\bfx,0,\xi^k_j)=1$, we obtain:
\[
1-1\wedge f(\bfx,\sqrt{\tau },\xi^k_j)\leq \sqrt{\tau } M_1(1+\|\xi^k_j\|^2). 
\]
Recall that $\xi^k_j$ is a sample of a standard normal variable
whose dimension is less than $q$ by Assumption \ref{aspt:sparse}
($q$ being the dimension of one block). 
Consequentially, $\E_{k,j}  (1-f(\bfx,\sqrt{\tau },\xi^k_j)\wedge 1)\leq  \sqrt{\tau } M_1(1+q)$, which leads to our first claim.

As for the second claim, given $\bfx,\bfz, \xi^k_j$, the probability of having only one proposal being accepted is $|f(\bfx,\sqrt{\tau },\xi^k_j)\wedge 1-f(\bfz,\sqrt{\tau },\xi^k_j)\wedge 1|$.  Let $\bfy_s$ be the linear interpolation between $\bfx$ and $\bfz$, so that $\bfy_0=\bfx$ and $\bfy_1=\bfz$. If $f(\bfx,\sqrt{\tau },\xi^k_j)$ and $f(\bfz,\sqrt{\tau },\xi^k_j)$ are both above $1$, then our claim holds trivially. 
Otherwise, either
$f(\bfx,\sqrt{\tau },\xi^k_j)$ or $f(\bfz,\sqrt{\tau },\xi^k_j)$ 
is less than $1$. Assume, without loss of generality, that $f(\bfx,\sqrt{\tau },\xi^k_j)<1$
and define 
\[
s^*=\begin{cases}\inf\{s\in [0,1]: f(\bfy_s,\sqrt{\tau },\xi^k_j)\geq 1\},\quad  &f(\bfz,\sqrt{\tau },\xi^k_j)>1;\\
1,\quad &\text{else}.
\end{cases}
\] 
Then, for $s\in [0,s^*]$, $f(\bfy_s,\sqrt{\tau },\xi^k_j)\leq 1$. Also note that, by Lemma \ref{lem:Lipacc} (claim(1)), 
\[
f(\bfx, \sqrt{\tau }, \xi^k_j)=f([\bfx_1,\ldots \bfx_m], \sqrt{\tau }, \xi^k_j)
\]
has dependence on $\bfx_i$  only if $i\in \calI_j$. 
Therefore, by Lemma \ref{lem:Lipacc} (2), there is a constant $M_2$, such that 
\begin{align*}
|f(\bfx,\sqrt{\tau },\xi^k_j)\wedge 1-&f(\bfz,\sqrt{\tau },\xi^k_j)\wedge 1|\leq |f(\bfy_0,\sqrt{\tau },\xi^k_j)-f(\bfy_{s^*},\sqrt{\tau },\xi^k_j)|\\
&=\left|\int^{s_*}_0   \nabla_{\bfy_s} f(\bfy_s,\sqrt{\tau },\xi^k_j) (\bfx^{k,j}-\bfz^{k,j})ds\right|\\
&\leq  \sqrt{\sum_{i\in\calI_j} \|\bfx^{k,j}_i-\bfz^{k,j}_i\|^2} \int^{s_*}_0   \|\nabla_{\bfy_s} f(\bfy_s,\sqrt{\tau },\xi^k_j)\|ds\\
&\leq \sqrt{\sum_{i\in\calI_j} \|\bfx^{k,j}_i-\bfz^{k,j}_i\|^2} \int^1_0  \frac{\|\nabla_{\bfy_s} f(\bfy_s,\sqrt{\tau },\xi^k_j)\|}{f(\bfy_s)}ds\\
&\leq\sqrt{\sum_{i\in\calI_j} \|\bfx^{k,j}_i-\bfz^{k,j}_i\|^2} M_2(\|\xi_j^k\|^2+1)\tau .
\end{align*}
Averaging the above expression over all possible outcomes of $\xi_j^k$
proves our second claim.

Similarly, for the third claim, we have
\begin{align*}
&\E_{k,j}|f(\bfx,\sqrt{\tau },\xi^k_j)\wedge 1-f(\bfz,\sqrt{\tau },\xi^k_j)\wedge 1| \|\xi^k_j\|\\
&\leq \E_{k,j}\sqrt{\sum_{i\in \calI_j} \|\bfx^{k,j}_i-\bfz^{k,j}_i\|^2} M_2(\|\xi_j^k\|^3+\|\xi_j^k\|)\tau.
\end{align*}
The upper bound in claim~(3) can be obtained by averaging over all $\xi_j^k$. 
\end{proof}

\subsection{Block-distance updates}
\label{sec:BlockUpdates}
We focus on the difference between the 
pair of coupled MALA-within-Gibbs iterations and define 
\begin{equation}
\label{tmp:Deltaremind}
\Delta^k_j:=\bfx^{k}_j-\bfz^{k}_j, \quad \Delta^{k,j}=\bfx^{k,j}-\bfz^{k,j},
\end{equation}
where each block of $\Delta^{k,j}$ is given by
\begin{equation*}
\Delta^{k,j}_i =\bfx^{k,j}_i-\bfz^{k,j}_i=\begin{cases}
\Delta^{k+1}_i,\quad &i\leq j-1;\\
\Delta^{k}_i,\quad &i\geq j.
\end{cases}
\end{equation*}
We first analyze the dynamics of $\|\Delta^k_j\|$
and use the results to prove Theorem \ref{thm:unbiased} 
under additional assumptions of block-wise log-concavity.
\begin{prop}
\label{prop:onestep}
Under Assumptions \ref{aspt:sparse} and \ref{aspt:regular},  there is a constant $M$ such that after the $j$-th MALA-within-Gibbs step at the $k$-th iteration, 
\begin{equation}
\label{eqn:proponeclaim}
\E_{k,j}(\|\Delta^{k+1}_j \|) \leq \|\Delta^{k}_j+\tau  \bfv_j (\bfx^{k,j})-\tau  \bfv_j (\bfz^{k,j})\|+M\tau ^{\frac 32}\sum_{i\in \calI_j} \|\Delta^{k,j}_i\|.
\end{equation}
\end{prop}

\begin{proof}
Recall that the coupling of the acceptance steps has four scenarios
(both accepted, both rejected, one rejected the other accepted). 
If both proposals are rejected, then 
\[
\|\Delta^{k+1}_j\|=\|\Delta^{k}_j\|.
\]
If both proposals are accepted, then 
\[
\|\Delta^{k+1}_j\|=\|\bfx^{k}_j+\tau \bfv_j (\bfx^{k,j})-\bfz^{k}_j-\tau \bfv_j (\bfz^{k,j})\|.
\] 
If only $\tilde{\bfx}^{k,j} $ is accepted, then
\begin{align*}
\|\Delta^{k+1}_j\|&=\|\bfx^k_{j}+\tau \bfv_j (\bfx^{k,j})+\sqrt{2\tau} \xi_{j}^k-\bfz^k_j\|\\
&\leq \|\Delta^{k}_j\|+\tau M_v+\sqrt{2\tau}\|\xi^k_j\|
\end{align*}
Likewise, if only $\bfztilde^k_{j} $ is accepted, then
\[
\|\Delta^{k+1}_j\|\leq  \|\Delta^{k}_j\|+\tau M_v+\sqrt{2\tau}\|\xi^k_j\|.
\]
Summing over all four scenarios, we obtain
\begin{align}
\notag
\E_{k,j} \|\Delta^{k+1}_j\|&\leq \|\Delta_j^k\| \Prob_{k,j}(\text{reject both})+\|\Delta^{k}_j+\tau \bfv_j (\bfx^{k,j})-\tau \bfv_j (\bfz^{k,j})\| \Prob_{k,j}(\text{accept both})\\
\notag
&\quad+ \E_{k,j}(\|\Delta^{k}_j\|+\tau M_v+\sqrt{2\tau}\|\xi^k_j\|)\mathbf{1}_\text{accept $\bfxtilde$ or $\bfztilde$}\\
\notag
&=\|\Delta^{k}_j+\tau \bfv_j (\bfx^{k,j})-\tau \bfv_j (\bfz^{k,j})\|\\
\notag
&\quad+(\|\Delta^k_j\|-\|\Delta^{k}_j+\tau \bfv_j (\bfx^{k,j})-\tau \bfv_j (\bfz^{k,j})\|)\Prob_{k,j}(\text{reject at least one})\\
\label{tmp:deltadecomp}
&\quad+\E_{k,j}(\tau M_v+\sqrt{2\tau}\|\xi^k_j\|)\mathbf{1}_\text{accept $\bfxtilde$ or $\bfztilde$}.
\end{align}
To prove the Proposition, it suffices to show that, for a constant $M$, the following two 
inequalities hold:
\begin{equation}
\label{tmp:decomp1} 
(\|\Delta^k_j\|-\|\Delta^{k}_j+\tau \bfv_j (\bfx^{k,j})-\tau \bfv_j (\bfz^{k,j})\|)\Prob_{k,j}(\text{reject at least one})
\leq \frac12 M\tau ^{\frac 32}\sum_{i\in \calI_j} \|\Delta^{k,j}_i\|,
\end{equation}
\begin{equation}
\label{tmp:decomp2} 
\E_{k,j}(\tau M_v+\sqrt{2\tau}\|\xi^k_j\|)\mathbf{1}_\text{accept $\bfxtilde$ or $\bfztilde$}
\leq \frac12 M\tau ^{\frac 32}\sum_{i\in \calI_j} \|\Delta^{k,j}_i\|.
\end{equation}
The reason is that, if the above inequalities hold,
we can use in \eqref{tmp:decomp1} and \eqref{tmp:decomp2} in $\eqref{tmp:deltadecomp}$, 
to obtain \eqref{eqn:proponeclaim}.

We will first show \eqref{tmp:decomp2}. Note that, by Lemma \ref{lem:couplingprob} claim(2), there is a constant $M_1$ such that 
\begin{equation}
\label{tmp:79}
\Prob_{k,j}(\text{accept $\bfxtilde$ or $\bfztilde$} )\leq M_1 \tau \sqrt{\sum_{i\in\calI_j} \|\bfx^{k,j}_i-\bfz^{k,j}_i\|^2}\leq 
 M_1 \tau \sum_{i\in\calI_j} \|\bfx^{k,j}_i-\bfz^{k,j}_i\|= M_1 \tau \sum_{i\in\calI_j} \|\Delta^{k,j}_i\|. 
\end{equation}
By Lemma \ref{lem:couplingprob} claim(3) 
\begin{equation}
\label{tmp:710}
\E_{k,j}\|\xi^k_j\|\mathbf{1}_\text{accept $\bfxtilde$ or $\bfztilde$}\leq M_1 \tau \sum_{i\in\calI_j} \|\bfx^{k,j}_i-\bfz^{k,j}_i\|= M_1 \tau \sum_{i\in\calI_j} \|\Delta^{k,j}_i\|.
\end{equation}
Since we assume that $\tau\leq1$, we can plug \eqref{tmp:79} and \eqref{tmp:710} into the left hand side of \eqref{tmp:decomp2}, and find a  constant $M_2$ so that 
\begin{align*}
\E_{k,j}(\tau M_v+\sqrt{2\tau}\|\xi^k_j\|)\mathbf{1}_\text{accept $\bfxtilde$ or $\bfztilde$}
&=M_v \tau \Prob(\text{accept $\bfxtilde$ or $\bfztilde$})+\sqrt{2\tau}\E_{k,j}\|\xi^k_j\|\mathbf{1}_\text{accept $\bfxtilde$ or $\bfztilde$}\\
&\leq \tau^{\frac32}M_2\sum_{i\in \calI_j} \|\bfx^{k,j}_i-\bfz^{k,j}_i\|= \tau^{\frac32}M_2\sum_{i\in \calI_j} \|\Delta^{k,j}_i\|.
\end{align*}
This proves \eqref{tmp:decomp2}.

We now use~\eqref{tmp:decomp2} to prove  \eqref{tmp:decomp1}. 
By the triangular inequality and Assumptions \ref{aspt:sparse} and \ref{aspt:regular},
we have that
\begin{equation}
\label{temp:item2}
\left|\|\Delta^{k}_j+\tau \bfv_j (\bfx^{k,j})-\tau \bfv_j (\bfz^{k,j})\|-\|\Delta^k_j\|\right|\leq \tau\|\bfv_j (\bfx^{k,j})-\bfv_j (\bfz^{k,j})\|\leq \tau H_v \sum_{i\in \calI_j} \|\Delta^{k,j}_i\|.
\end{equation}
Moreover, by Lemma \ref{lem:couplingprob} (1), there is a constant $M_3$ such that 
\begin{equation}
\label{tmp:711}
\Prob_{k,j}(\text{reject at least one})\leq M_3\sqrt{\tau}.
\end{equation}
The product of  \eqref{temp:item2} and \eqref{tmp:711} leads to \eqref{tmp:decomp1}:
\begin{align*}
(\|\Delta^k_j\|-\|\Delta^{k}_j+\tau \bfv_j (\bfx^{k,j})-\tau \bfv_j (\bfz^{k,j})\|)\Prob_{k,j}(\text{reject at least one})\leq H_vM_3\tau^{\frac32} \sum_{i\in \calI_j} \|\Delta^{k,j}_i\|.
\end{align*}
This concludes our proof.
\end{proof}

\subsection{Contraction with block-wise log-concavity (Proof of Theorem~\ref{thm:unbiased})}
\label{sec:ProofOfTheorem}
We complete the proof of Theorem~\ref{thm:unbiased}
by showing that the assumption of block-log-concavity
implies that the coupled pair of MALA-within-Gibbs defines a contraction.
\begin{lem}
\label{lem:sparseHess}
Under Assumptions \ref{aspt:sparse} and \ref{aspt:regular}, for all coupled pairs $\bfx,\bfz$,
and independently of the iteration number $k$ (which we drop for convenience):
\[
\|\nabla^2_{\bfx_i,\bfx_j}\log \pi (\bfx)-\nabla^2_{\bfz_i,\bfz_j}\log \pi (\bfz)\|=
\|\nabla_{\bfx_i}\bfv_j (\bfx)-\nabla_{\bfz_i}\bfv_j (\bfz)\|
\leq H_v\sum_{l\in \calI_{i,j}}  \|\bfx_l-\bfz_l\|,
\]
where $\calI_{i,j}=\calI_i\cap\calI_j$  has cardinality $|\calI_{i,j}|\leq S$.
\end{lem}
\begin{proof}
By Assumption \ref{aspt:sparse}, we know that $\bfv_j(\bfx)$ has no dependence on $\bfx_i$ if $i\notin \calI_j$, so that we can write $\bfv_j(\bfx_{\calI_j})$.
Similarly, $\nabla_{\bfx_i}\bfv_j(\bfx)$ can be written as $\nabla_{\bfx_i}\bfv_j(\bfx_{\calI_j})$. 
For any $\bfx$ and $\bfz$, pick $\bfy,\bfu\in\reals^d$ so that $\bfy_{\calI_j}=\bfx_{\calI_j}$, $\bfy_{\calI_j^c}=\bfz_{\calI_j^c}$ and $\bfu_{\calI_i}=\bfy_{\calI_i}$, $\bfu_{\calI_{i}^c}=\bfz_{\calI_{i}^c}$. Note that $\bfu$ will differ from $\bfz$ only at the blocks with indices in $\calI_{i,j}$. Then, since  $\nabla_{\bfx_i}\bfv_j (\bfx)=[\nabla_{\bfx_j}\bfv_i (\bfx)]^T$ and by Assumption \ref{aspt:regular},
\begin{align*}
\|\nabla_{\bfx_i}\bfv_j (\bfx)-\nabla_{\bfz_i}\bfv_j (\bfz)\|&=\|\nabla_{\bfy_i}\bfv_j (\bfy)-\nabla_{\bfz_i}\bfv_j (\bfz)\|\\
&=\|\nabla_{\bfy_j}\bfv_i (\bfy)-\nabla_{\bfz_j}\bfv_i (\bfz)\|\\
&=\|\nabla_{\bfu_j}\bfv_i (\bfu)-\nabla_{\bfz_j}\bfv_i (\bfz)\|\\
&\leq H_v\|\bfu-\bfz\|\leq H_v\sum_{l\in \calI_{i,j}} \|\bfu_l-\bfz_l\|=H_v\sum_{l\in \calI_{i,j}} \|\bfx_l-\bfz_l\|. 
\end{align*}
\end{proof}

\label{sec:proofoftheorem}
\begin{proof}[Proof of Theorem \ref{thm:unbiased}]
From Proposition \ref{prop:onestep}, 
and using the fact that the differences of the pair of coupled MALA-within-Gibbs iterates in their blocks at the $j$-th block update is given by \eqref{tmp:Deltaremind}, we have a.s.,
\begin{equation}
\label{tmp:prop41}
\E_{k} \|\Delta_j^{k+1}\|\leq \E_{k}\|\Delta_j^k+\tau \bfv_j (\bfx^{k,j})-\tau \bfv_j(\bfz^{k,j})\|+M\tau^{\frac32}\sum_{i<j,i\in \calI_j} \E_{k}\|\Delta^{k+1}_i\|+M\tau^{\frac32}\sum_{i\geq j,i\in \calI_j} \|\Delta^{k}_i\|. 
\end{equation}
Here, recall that $\E_k$ is the conditional expectation with respect to the information available before the $k$-th Gibbs cycle. To simplify notation, we define, for an arbitrary function $g$, $g[\bfx,\bfz]=\int^1_0 g(s\bfx+(1-s)\bfz)ds$.  Recall that $\Delta^{k,j}=\bfx^{k,j}-\bfz^{k,j}=[\Delta_1^{k+1}, \Delta_2^{k+1},\cdots,\Delta_{j-1}^{k+1},\Delta_j^k,\cdots, \Delta_m^k]$. Also, we use 
\[
\nabla^2_{\bfx_j,\bfx} \log \pi(\bfx)=[\nabla^2_{\bfx_j,\bfx_1} \log \pi(\bfx),\ldots, \nabla^2_{\bfx_j,\bfx_m} \log \pi(\bfx)],
\] to denote the $j$-th block-wise row of the  Hessian log density.

\textbf{First step:} we decompose the first term 
of the right hand side of \eqref{tmp:prop41} into terms involving block-wise quantities. 
To this end, we first decompose
\begin{align}
\notag
\bfv_j (\bfx^{k,j})- \bfv_j(\bfz^{k,j})&=\nabla_{\bfx_j}\log \pi(\bfx^{k,j})-\nabla_{\bfz_j} \log \pi(\bfz^{k,j})\\
\notag
&=\int^{1}_0\nabla^2_{\bfx_j,\bfx} \log \pi(\bfx^{k,j}+s \Delta^{k,j}) \Delta^{k,j} ds\\
\notag
&=(\nabla^2_{\bfx_j,\bfx} \log \pi)[\bfx^{k,j}, \bfz^{k,j}] \Delta^{k,j} \\
\notag
&=(\nabla^2_{\bfx_j,\bfx_j} \log \pi)[\bfx^{k,j},\bfz^{k,j}] \Delta^{k}_j+\sum_{i<j,i\in \calI_j} (\nabla^2_{\bfx_j,\bfx_i} \log \pi)[\bfx^{k,j}, \bfz^{k,j}] \Delta^{k+1}_i \\
\label{tmp:vdiff}
&\quad+ \sum_{i>j,i\in \calI_J} (\nabla_{\bfx_j,\bfx_i}^2 \log \pi) [\bfx^{k,j},\bfz^{k,j}] \Delta^{k}_i\\
\notag
&=(\nabla^2_{\bfx_j,\bfx_j} \log \pi)[\bfx^{k,j},\bfz^{k,j}] \Delta^{k}_j+\sum_{i\neq j, i\in \calI_j} (\nabla^2_{\bfx_j,\bfx_i} \log \pi)[\bfx^{k,j}, \bfz^{k,j}] \Delta^{k}_i\\
\notag
&\quad + \sum_{i<j,i\in \calI_J} (\nabla_{\bfx_j,\bfx_i}^2 \log \pi) [\bfx^{k,j},\bfz^{k,j}] (\Delta^{k+1}_i-\Delta^{k}_i).
\end{align}
Therefore, we can bound the first term on the right hand side of \eqref{tmp:prop41} by
\begin{align}
\label{tmp:prop41first1}
\E_k\|\Delta_j^k+\tau \bfv_j (\bfx^{k,j})-\tau \bfv_j(\bfz^{k,j})\|\leq \calD_1+\calD_2+\calD_3,
\end{align}
where
\begin{align*}
\calD_1:&=\E_k\left\|\Delta^{k}_j+\tau (\nabla_{\bfx_j,\bfx_j}^2 \log \pi)[\bfx^{k,j},\bfz^{k,j}] \Delta^{k}_j \right\|,\\
\calD_2:&=\tau\sum_{i\neq j, i\in \calI_j}\E_k\left\|(\nabla^2_{\bfx_j,\bfx_i} \log \pi)[\bfx^{k,j}, \bfz^{k,j}] \right\|\| \Delta^{k}_i\|,\\
\calD_3:&=\tau\sum_{i<j,i\in \calI_j}\E_k\left\|(\nabla^2_{\bfx_j,\bfx_i} \log \pi)[\bfx^{k,j}, \bfz^{k,j}] \right\|\|\Delta^{k+1}_i -\Delta^{k}_i\|.
\end{align*}

\textbf{Second step:} we replace $[\bfx^{k,j},\bfz^{k,j}] $ in $\calD_1$ and $\calD_2$ by $[\bfx^{k},\bfz^{k}] $. The reason for doing this is that $\bfx^k$ and $\bfz^k$ are fixed for all $j$, which makes these quantities easier to analyze.  Under the Lipschitz conditions for $\nabla_{\bfx_j,\bfx_i}\log\pi$ in Assumption \ref{aspt:blockconcave}, and recalling that 
\[
\bfx^{k,j}=[\bfx^{k+1}_1,\ldots,\bfx^{k+1}_{j-1}, \bfx^{k}_j,\ldots, \bfx^{k}_m],
\]
we can bound $\|(\nabla^2_{\bfx_j,\bfx_i}\log\pi)[\bfx^{k,j},\bfz^{k,j}]-(\nabla^2_{\bfx_j,\bfx_i}\log\pi)[\bfx^k,\bfz^k]\|$ using Lemma \ref{lem:sparseHess}:
\begin{align}
\notag
&\|(\nabla^2_{\bfx_j,\bfx_i}\log\pi)[\bfx^{k,j},\bfz^{k,j}]-(\nabla^2_{\bfx_j,\bfx_i}\log\pi)[\bfx^k,\bfz^k]\|\\
\notag
&\leq \int^1_0 \|\nabla^2_{\bfx_j,\bfx_i}\log\pi(s\bfx^{k,j}+(1-s)\bfz^{k,j})- \nabla^2_{\bfx_j,\bfx_i}\log\pi (s\bfx^k+(1-s)\bfz^k)\|ds\\
\notag
&\leq \int^1_0 H_v\left(\sum_{l\in \calI_{j,i},l<j} s\|\bfx^{k+1}_l-\bfx^k_l\|+(1-s) \|\bfz^{k+1}_l-\bfz^k_l\|\right)ds\\
\label{tmp:Hessdiff}
&\leq \frac12 H_v  \sum_{l\in \calI_{j,i},l<j}  (\|\bfx^{k+1}_l-\bfx^{k}_l\|+ \|\bfz^{k+1}_l-\bfz^k_l\|).
\end{align}
In \eqref{tmp:Hessdiff}, first note that the summation is over at most $S$ terms, because the cardinality of the index set satisfies 
\[
|\{l: l\in \calI_{j,i}, l<j\}|\leq |\calI_{j,i}|\leq S.
\] 
Second, note that either $\bfx^{k+1}_l=\bfx^k_l$ or $\bfx^{k+1}_l=\tilde{\bfx}^k_l=\bfx^k_l+\tau\bfv_l(\bfx^{k,l})+\sqrt{2\tau}\xi^{k}_l$, depending whether the proposal is rejected or not. Therefore,
\begin{equation}
\label{tmp:xdiff}
\E_{k}\|\bfx^k_l-\bfx^{k+1}_l\|\leq \sqrt{\E_{k} \|\bfx^k_l-\tilde{\bfx}^k_l\|^2}\leq \sqrt{\tau^2 \E_{k} \|\bfv_l(\bfx^{k,l})\|^2+2\tau q}\leq \tau M_v+\tau^{\tfrac12} \sqrt{2q}.
\end{equation}
The above bound also applies to  $\E_{k}\|\bfz^k_l-\bfz^{k+1}_l\|$, for the same reasons. Therefore, by \eqref{tmp:Hessdiff}, for the constant $M_0:= H_vS(M_v\sqrt{\tau}+\sqrt{2q})$, 
\begin{equation*}
\E_k \tau \|(\nabla^2_{\bfx_j,\bfx_i}\log\pi)[\bfx^{k,j},\bfz^{k,j}]-(\nabla^2_{\bfx_j,\bfx_i}\log\pi)[\bfx^k,\bfz^k]\|\leq H_v S(M_v\tau^2+\sqrt{2q}\tau^{\frac32})\leq M_0\tau^\frac{3}{2}.
\end{equation*}
This leads to following bound for $\calD_1$ in \eqref{tmp:prop41first1},
\begin{align}
\notag
&\E_k\left\|\Delta^{k}_j+\tau (\nabla_{\bfx_j,\bfx_j}^2 \log \pi)[\bfx^{k,j},\bfz^{k,j}] \Delta^{k}_j \right\|\\
\notag
&\leq \left\|\Delta^{k}_j+\tau (\nabla_{\bfx_j,\bfx_j}^2 \log \pi)[\bfx^{k},\bfz^{k}] \Delta^{k}_j \right\|+\tau \E_k \|(\nabla^2_{\bfx_j,\bfx_i}\log\pi)[\bfx^{k,j},\bfz^{k,j}]-(\nabla^2_{\bfx_j,\bfx_i}\log\pi)[\bfx^k,\bfz^k]\|\|\Delta_j^k\|\\
\notag
&\leq \left\|\Delta^{k}_j+\tau (\nabla^2_{\bfx_j,\bfx_j} \log \pi)[\bfx^{k},\bfz^{k}] \Delta^{k}_j \right\|+M_0\tau^{\frac32}\|\Delta_j^k\|\\
\label{tmp:Ekdiag}
&\leq (1+\tau H_{j,j}[\bfx^{k},\bfz^{k}]+M_0\tau^{\frac32})\|\Delta_j^k\|. 
\end{align}
To obtain the last inequality, we use Assumption \ref{aspt:blockconcave} i).  

Likewise, we can build an upper bound for $\calD_2$ in \eqref{tmp:prop41first1}. 
\begin{align}
\notag
&\tau \E_k\left\|(\nabla^2_{\bfx_j,\bfx_i} \log \pi)[\bfx^{k,j}, \bfz^{k,j}]\right\|\| \Delta^{k}_i \|\\
\notag
&\leq \tau\E_k\left\|(\nabla^2_{\bfx_j,\bfx_i} \log \pi)[\bfx^{k}, \bfz^{k}]\right\|\| \Delta^{k}_i \|
+\tau\E_k \|(\nabla^2_{\bfx_j,\bfx_i}\log\pi)[\bfx^{k,j},\bfz^{k,j}]-(\nabla^2_{\bfx_j,\bfx_i}\log\pi)[\bfx^k,\bfz^k]\|\| \Delta^{k}_i \|\\
\notag
&\leq (\tau H_{j,i} [\bfx^{k},\bfz^{k}]+M_0\tau^{\frac{3}{2}})\|\Delta_i^{k}\|. 
\end{align}
In summary
\begin{equation}
\label{tmp:Ekoff}
\calD_2=\sum_{i\neq j, i\in \calI_j}\E_k\left\|(\nabla^2_{\bfx_j,\bfx_i} \log \pi)[\bfx^{k,j}, \bfz^{k,j}]\right\|\| \Delta^{k}_i \|\leq S(\tau H_{j,i} [\bfx^{k},\bfz^{k}]+M_0\tau^{\frac{3}{2}})\|\Delta_i^{k}\|. 
\end{equation}

\textbf{Third step:}
we bound $\calD_3$ in \eqref{tmp:prop41first1}.
Note that by Assumption \ref{aspt:regular}, $\|\nabla^2_{\bfx_i,\bfx_j}\log \pi \|\leq H_v$, 
so that 
\begin{equation}
\label{tmp:thirdstep}
\calD_3\leq  \tau H_v \sum_{i\in I_j}\E_k \|\Delta_i^{k+1}-\Delta_i^k\|=\tau H_v \sum_{i\in I_j}\E_k \|(\bfx^{k+1}_i-\bfx^{k}_i)-(\bfz^{k+1}_i-\bfz^{k}_i)\|.
\end{equation}
Note that $\bfx^{k+1}_i$ is either $\bfx^k_i$ or $\tilde{\bfx}^k_i$ based on whether $\tilde{\bfx}^k_i$ is rejected or not.
Thus,
\[
\bfx^{k+1}_i-\bfx^{k}_i=\unit_{\tilde{\bfx}^k_i\,\text{is accepted} } (\tau \bfv_i(\bfx^{k,i})+\sqrt{2\tau}\xi^k_i),\quad
\bfz^{k+1}_i-\bfz^{k}_i=\unit_{\tilde{\bfz}^k_i\,\text{is accepted} } (\tau\bfv_i(\bfz^{k,i})+\sqrt{2\tau}\xi^k_i).
\]
This leads to 
\begin{equation}
\E_k \|(\bfx^{k+1}_i-\bfx^{k}_i)-(\bfz^{k+1}_i-\bfz^{k}_i)\|=
\calC_1+\calC_2+\calC_3,
\label{tmp:gapdiff1}
\end{equation}
where
\begin{align*}
\calC_1:&=\E_k \tau\|\unit_{\tilde{\bfx}^k_i\,\text{and}\,\tilde{\bfz}^k_i\text{are accepted} }(\bfv_i(\bfx^{k,i})-\bfv_i(\bfz^{k,i})) \|,\\
\calC_2:&=\E_k \|\unit_{\tilde{\bfx}^k_i\,\text{is accepted}\,\tilde{\bfz}^k_i\text{is rejected} }(\tau\bfv_i(\bfx^{k,i})+\sqrt{2\tau}\xi^k_i) \|,\\
\calC_3:&=\E_k \|\unit_{\tilde{\bfz}^k_i\,\text{is accepted}\,\tilde{\bfx}^k_i\text{is rejected} }(\tau\bfv_i(\bfz^{k,i})+\sqrt{2\tau}\xi^k_i) \|.
\end{align*}
To bound $\calC_1$, 
we note that by \eqref{tmp:vdiff} and Assumption \ref{aspt:regular}, we have
\[
\|\bfv_i(\bfx^{k,i})-\bfv_i(\bfz^{k,i})\|\leq H_v \left(\sum_{l\geq i, l\in \calI_i} \|\Delta^k_l\|+\sum_{l<i, l\in \calI_i} \|\Delta^{k+1}_l\|\right).
\]
Consequentially,
\begin{equation}
\label{tmp:723}
\calC_1\leq \E_k \tau\|\bfv_i(\bfx^{k,i})-\bfv_i(\bfz^{k,i}) \| \leq \tau H_v\left(\sum_{l\geq i, l\in \calI_i} \|\Delta^k_l\|+\sum_{l<i, l\in \calI_i} \E_k\|\Delta^{k+1}_l\|\right).
\end{equation}
Next, we bound $\calC_2$
using Lemma \ref{lem:couplingprob}, claims (2) and (3), so that for some constant $M_0'$
\begin{align*}
\calC_2 &\leq \tau\E_k\|\unit_{\tilde{\bfx}^k_i\,\text{is accepted}\,\tilde{\bfz}^k_i\text{is rejected} }\bfv_i(\bfx^{k,i})\|+\sqrt{2\tau}\E_k \|\unit_{\tilde{\bfx}^k_i\,\text{is accepted}\,\tilde{\bfz}^k_i\text{is rejected} }\xi^{k}_i\|\\
&\leq \tau M_v\Prob_k(\text{accept only one of }\tilde{\bfx}^k_i,\tilde{\bfz}^k_i )+\sqrt{2\tau}\E_k \unit_{\text{accept only one of }\tilde{\bfx}^k_i,\tilde{\bfz}^k_i }\|\xi^{k}_i\|\\
&\leq  M_0' (M_v\tau^2+\sqrt{2\tau}\tau) \E_k\sqrt{\sum_{l\in \calI_i} \|\bfx^{k,i}_l-\bfz^{k,i}_l\|^2}.
\end{align*}
Also note that 
\begin{align*}
\E_k\sqrt{\sum_{l\in \calI_i} \|\bfx^{k,i}_l-\bfz^{k,i}_l\|^2}\leq \E_k\sum_{l\in \calI_i} \|\bfx^{k,i}_l-\bfz^{k,i}_l\|=\sum_{l\geq i, l\in \calI_i} \|\Delta^k_l\|+\sum_{l<i, l\in \calI_i} \E_k\|\Delta^{k+1}_l\|. 
\end{align*}
Therefore, 
\begin{equation}
\label{tmp:726}
\calC_2 \leq M_0'(M_v\tau+\sqrt{2\tau})\left(\sum_{l\geq i, l\in \calI_i} \|\Delta^k_l\|+\sum_{l<i, l\in \calI_i} \E_k\|\Delta^{k+1}_l\|\right). 
\end{equation}
The same bound holds for $\calC_3$.
Combining \eqref{tmp:thirdstep}, \eqref{tmp:gapdiff1}, \eqref{tmp:723} and \eqref{tmp:726}, we find that, for some constant $M_1$, 
\begin{equation}
\label{tmp:thirdfinish}
\calD_3\leq  \sum_{i\in \calI_j} \tau H_v (\calC_1+\calC_2+\calC_3 )\leq S M_1 \tau^{\frac32}\left(\sum_{l\in \calI_j^2} \|\Delta^k_l\|+\sum_{l\in \calI_j^2} \E_k\|\Delta^{k+1}_l\|\right).
\end{equation}
Here $\calI_j^2:=\bigcup_{i\in \calI_j} \calI_i$; by the union bound of cardinality, $|\calI_j^2|\leq S^2$. Also note that because $j\in \calI_j$ so that $\calI_j\subset \calI_j^2$.

\textbf{Summary of the first three steps:} we bound the first term 
of the right hand side of \eqref{tmp:prop41}  by bounding $\calD_1$ with \eqref{tmp:Ekdiag}, $\calD_2$ with \eqref{tmp:Ekoff}, and $\calD_3$ with \eqref{tmp:thirdfinish}
In conclusion, 
$\E_k  \|\Delta_j^{k+1}\|$ 
can be bounded by 
\begin{align*}
&\E_k  \|\Delta_j^{k+1}\| \leq (1+\tau H_{j,j}[\bfx^{k},\bfz^{k}]+M_0 \tau^{\frac32})\|\Delta^k_j\|+\sum_{i\neq j,i\in \calI_j} (\tau H_{j,i}[\bfx^{k},\bfz^{k}]+M_0\tau^{\frac32}) \|\Delta^{k}_i\|\\
&\quad+\tau^{\frac32}\left(\sum_{i\in \calI^2_j}S M_1  \E_k \|\Delta^{k}_i\|+\sum_{i\in \calI_j^2}S M_1 \E_k \|\Delta^{k+1}_i\|+M\sum_{i<j,i\in \calI_j} \E_{k}\|\Delta^{k+1}_i\|+M\sum_{i\geq j,i\in \calI_j} \|\Delta^{k}_i\|\right),
\end{align*}
or, in a more succinct form,
\begin{align}
\notag
\E_k  \|\Delta_j^{k+1}\|&\leq (1+\tau H_{j,j}[\bfx^{k},\bfz^{k}]+M_0 \tau^{\frac32})\|\Delta^k_j\|+\sum_{i\neq j, i\in \calI_j} (\tau H_{j,i}[\bfx^{k},\bfz^{k}]+M_0\tau^{\frac32}) \|\Delta^{k}_i\|,\\
\label{tmp:Deltafinal}
&\qquad\quad+\sum_{i\in \calI^2_j}(S M_1+M)\tau^{\frac32}  \E_k \|\Delta^{k}_i\|+\sum_{i\in \calI_j^2}(S M_1+M)\tau^{\frac32}  \E_k \|\Delta^{k+1}_i\|. 
\end{align}
This provides us upper bounds on how $\|\Delta_j^{k}\|$ 
changes when the iteration number $k$ increases. 

\textbf{Fourth step:} we first rewrite \eqref{tmp:Deltafinal} in a more compact matrix formulation. 
Define the following $\reals^{m\times m}$ matrices by their entries 
\[
[\bfM_\tau ]_{j,i}=\tau^{\frac{3}{2}}(SM_1+M)1_{i\in \calI^2_j},\quad  [\bfM^0_\tau ]_{j,i}=M_0 \tau^\frac{3}{2} 1_{i\in \calI_j}. 
\]
Next, we define $D_k:=[ \|\Delta^k_1\|,\|\Delta^k_2\|,\cdots, \|\Delta^k_m\|]^T$. It can be verified that 
\eqref{tmp:Deltafinal} is equivalent to the $j$-th row of the following vector inequality, 
\begin{equation}
\label{tmp:Dk}
(\bfI-\bfM_\tau) \E_k D_{k+1}\preceq (\bfI+\bfH[\bfx^k,\bfz^k]\tau+\bfM_\tau+\bfM_\tau^0) D_k,
\end{equation}
where we use the notation $\bfx\preceq \bfy$ to mean that 
the two vectors $\bfx$ and $\bfy$ satisfy $\bfx_i\leq \bfy_i$ for all blocks indexed by $i=1,\dots,m$. 
Recall that for a matrix, its $l_2$-operator norm is bounded by its $l_1$-norm, see \cite{MTM19}. 
Thus,
\begin{equation}
\label{tmp:A2}
\|\bfM_\tau\|\leq \tau^{\frac{3}{2}}\max_j |\calI_j^2| SM_1 \leq \tau^{\frac{3}{2}} (M_1S+M)S^2,\quad
\|\bfM^0_\tau\|\leq \tau^{\frac{3}{2}}\max_j |\calI_j| M_0  \leq \tau^{\frac{3}{2}} M_0S. 
\end{equation}
For sufficiently small $\tau$, $\|\bfM_\tau\|<1$, and we can thus write
\[
(\bfI-\bfM_\tau)^{-1}=\bfI+\bfM_\tau +(\bfM_\tau)^2+(\bfM_{\tau})^3+\cdots,
\]
and 
\begin{equation}
\label{tmp:A3}
\|(\bfI-\bfM_\tau)^{-1}\|\leq \sum_{i=0}^\infty\|\bfM_\tau\|^i=\frac{1}{1-\tau^{\frac{3}{2}} (M_1S+M)S^2}.
\end{equation}
Note that all entries of $\bfM_\tau$ are positive, therefore all entries of $(\bfM_{\tau})^n$ and $(\bfI-\bfM_\tau)^{-1}$ are positive as well. Therefore, \eqref{tmp:Dk} leads to 
\[
 \E_k D_{k+1}\preceq  (\bfI-\bfM_\tau)^{-1}  (\bfI+\bfH[\bfx^k,\bfz^k]\tau+\bfM_\tau+\bfM_\tau^0)D_k. 
\]
By law of total expectation, we have 
\[
\E D_k\preceq \E \prod_{l=0}^{k-1}(\bfI-\bfM_\tau)^{-1}  (\bfI+\bfH[\bfx^l,\bfz^l]\tau+\bfM_\tau+\bfM_\tau^0) D_0 .
\]
Finally we provide an a.s. upper bound for $\|(\bfI-\bfM_\tau)^{-1}  (\bfI+\bfH[\bfx^l,\bfz^l]\tau+\bfM_\tau+\bfM_\tau^0)\|$.
By the log-concavity assumption, $\|\bfI+\bfH [\bfx^{k},\bfz^{k}]\tau\|\leq 1-\tau \lambda_{H}$. 
Thus, by  \eqref{tmp:A2} and \eqref{tmp:A3}, for sufficiently small $\tau$, 
\[
\|(\bfI-\bfM_\tau)^{-1}  (\bfI+\bfH[\bfx^l,\bfz^l]\tau+\bfM_\tau+\bfM_\tau^0)\|\leq \frac{1-\tau \lambda_H +\tau^{\frac{3}{2}}(M_0 S+M_1S^3+MS^2)}{1-\tau^{\frac{3}{2}}(M_1S+M)S^2}\leq (1-(1-\delta)\tau \lambda_H),\quad a.s..
\] 
Finally, we have that 
\[
\|\E D_k\|\preceq \left\|\E \prod_{k=0}^{k-1}(\bfI-\bfM_\tau)^{-1}  (\bfI+\bfH[\bfx^l,\bfz^l]\tau+\bfM_\tau+\bfM_\tau^0) D_0 \right\|\leq  (1-(1-\delta)\tau \lambda_H)^k  \|\E D_0\|. 
\]
\end{proof}

\end{document}